\documentclass[11pt,letterpaper,english,onecolumn,draftclsnofoot]{IEEEtran}
\usepackage[T1]{fontenc}
\usepackage[latin9]{inputenc}
\usepackage{color}
\usepackage{babel}

\usepackage{verbatim}
\usepackage{amsmath}
\usepackage{graphicx}
\usepackage{amssymb}
\usepackage[unicode=true, pdfusetitle,
 bookmarks=true,bookmarksnumbered=false,bookmarksopen=false,
 breaklinks=false,pdfborder={0 0 1},backref=false,colorlinks=false]
 {hyperref}

\makeatletter

\providecommand{\tabularnewline}{\\}

\let\oldforeign@language\foreign@language
\DeclareRobustCommand{\foreign@language}[1]{%
  \lowercase{\oldforeign@language{#1}}}
\newtheorem{definitn}{Definition}
\newtheorem{lemma}{Lemma}
\newtheorem{thm}{Theorem}
\newtheorem{cor}{Corollary}
\newtheorem{remrk}{Remark}
\newtheorem{example}{Example}

\usepackage{cite}
\usepackage{times}
\usepackage{hyperref}

\makeatother

\begin{document}

\title{\textcolor{black}{Technical Report: MIMO B-MAC Interference Network
Optimization under Rate Constraints by Polite Water-filling and Duality}}

\author{\textcolor{black}{\normalsize An Liu, Youjian (Eugene) Liu, Haige
Xiang, Wu Luo}%
\thanks{\textcolor{black}{The work was supported in part by NSFC Grant No.60972008,
and in part by US-NSF Grant CCF-0728955 and ECCS-0725915. An Liu (Email:
wendaol@pku.edu.cn), Haige Xiang, and Wu Luo are with the State Key
Laboratory of Advanced Optical Communication Systems \& Networks,
School of EECS, Peking University, China. Youjian Liu is with the
Department of Electrical, Computer, and Energy Engineering, University
of Colorado at Boulder, USA. The corresponding author is Wu Luo.}%
}}

\maketitle
\vspace{-0.5in}

\begin{abstract}
\textcolor{black}{We take two new approaches to design efficient algorithms
for transmitter optimization under rate constraints to guarantee the
Quality of Service in }general MIMO interference networks, named B-MAC
Networks, which is a combination of multiple interfering broadcast
channels (BC) and multiaccess channels (MAC). Two related optimization
problems, maximizing the minimum of weighted rates under a sum-power
constraint and minimizing the sum-power under rate constraints, are
considered. The first approach takes advantage of existing efficient
algorithms for SINR problems by building a bridge between rate and
SINR through the design of optimal mappings between them so that the
problems can be converted to SINR constraint problems. The approach
can be applied to other optimization problems as well. The second
approach employs polite water-filling, which is the optimal network
version of water-filling that we recently found. It replaces almost
all generic optimization algorithms currently used for networks and
reduces the complexity while demonstrating superior performance even
in non-convex cases. Both centralized and distributed algorithms are
designed and the performance is analyzed in addition to numeric examples.\end{abstract}

\markboth{Peking University and University of Colorado at Boulder Joint Technical
Report, June $28^{\textrm{TH}}$, 2010}{and this is for right pages}

\newpage{}

\section{\textcolor{black}{Introduction}}

\subsection{System Setup and Problem Statement}

\textcolor{black}{We study t}he optimization under rate constraints
for general multiple-input multiple-output (MIMO) interference networks,\textcolor{black}{{}
named MIMO B-MAC networks \cite{Liu_IT10s_Duality_BMAC},} where each
transmitter may send data to multiple receivers and each receiver
may collect data from multiple transmitters. Consequently, the network
is a combination of multiple interfering broadcast channels (BC) and
multiaccess channels (MAC).\textcolor{black}{{} It includes BC, MAC,
interference channels, X channels }\cite{Maddah-Al_IT_MMKforX,Jafar_IT08_DOFMIMOX}\textcolor{black}{,
X networks }\cite{Jafar_09IT_DOFXchannel}\textcolor{black}{{} and most
practical wireless networks, such as cellular networks, WiFi networks,
DSL,} as special cases\textcolor{black}{.} We assume Gaussian input
and that each interference is either completely cancelled or treated
as noise. A wide range of interference cancellation is allowed, from
no cancellation to any cancellation specified by a valid binary \textit{coupling
matrix} of the data links. For example, simple linear receivers, dirty
paper coding (DPC) \cite{Costa_IT83_Dirty_paper} at transmitters,
and/or successive interference cancellation (SIC) at receivers may
be employed.

Two optimization problems are considered for \textcolor{black}{guaranteeing
the Quality of Service (QoS), where each data link has a target rate}.
\textcolor{black}{The feasibility of the target rates can be solved
by a feasibility optimization problem (}\textbf{FOP}\textcolor{black}{)}
which maximizes the minimum of scaled rates of all links, where the
scale factors are the inverse of the target rates. A\textcolor{black}{ll
target rates can be achieved simultaneously if and only if the optimum
of }\textbf{\textcolor{black}{FOP}}\textcolor{black}{{} is greater than
or equal to one.}\textbf{\textcolor{black}{{} FOP}}\textcolor{black}{{}
can be used in admission control. If the target rates are feasible,
the system tries to operate at minimum total transmission power in
order to prolong total battery life and to reduce the total interference
to other networks by }solving the \textcolor{black}{sum power minimization
problem }(\textbf{SPMP}) under the rate constraints\textcolor{black}{.}

\textcolor{black}{We study both centralized and distributed optimizations.
The centralized optimization with global channel state information
(CSI) provides an upper bound of the performance and a stepping stone
}to the design of the \textcolor{black}{distributed optimization algorithms.}
In some cases such as cooperative cellular networks, it is possible
to obtain global CSI if the base stations are allowed to exchange
CSI, making the \textcolor{black}{centralized optimization }relevant.
In ad hoc or large networks, we have to design distributed optimization
algorithms with local CSI.

\subsection{Related Works\label{sub:Soltuions-for-MAC}}

\textcolor{black}{The SINR version of }\textbf{\textcolor{black}{FOP}}\textcolor{black}{{}
and }\textbf{\textcolor{black}{SPMP}}\textcolor{black}{{} under SINR
constraints have been well studied for various cases, e.g., \cite{Madhow_VTC99_LOSsinrdual,Chang_TWC02_BFalgduality,Martin_ITV_04_BFdual,Boche_05TSP_MISOSINR,Rao_TOC07_netduality}
using SINR duality }\cite{Rashid:98,Visotski__VTC99_SIMODual,Boche_VTC02_GenDual_BF,Tse_IT03_MIMO_broadcast}\textcolor{black}{,
}which means that if a set of SINRs is achievable in the forward links,
then the same SINRs can be achieved in the reverse links when the
set of transmit and receive beamforming vectors are fixed. Thus, optimizing
the transmit vectors of the forward links is equivalent to the much
simpler problem of optimizing the receive vectors in the reverse links.\textcolor{black}{{}
However, these algorithms lack the following. 1) They cannot be directly
used to solve }\textbf{\textcolor{black}{FOP}}\textcolor{black}{{} and
}\textbf{\textcolor{black}{SPMP}}\textcolor{black}{{} under rate constraints
because the optimal number of beams for each link and the power/rate
allocation over these beams are unknown; 2) Except for \cite{Boche_05TSP_MISOSINR},
interference cancellation is not considered; 3) The optimal encoding
and decoding order when interference cancellation is employed is not
solved. }

Considering interference cancellation and encoding/decoding order,
the \textbf{FOP} and \textbf{SPMP} for MIMO BC/MAC have been completely
solved in \cite{Nihar_06_ISIT_SymCapMIMODown} by converting them
to convex weighted sum-rate maximization problems for MAC. %
{}\textcolor{black}{The complexity is very high because the }steepest
ascent algorithm for the weighted sum-rate maximization\textcolor{black}{{}
needs to be solved repeatedly for each weight vector searched by the
ellipsoid algorithm.} A high complexity algorithm that can find the
optimal encoding/decoding order for MISO BC/SIMO MAC is proposed in
\cite{Yu_07TOC_MACorder} that needs several inner and outer iterations.
A heuristic low-complexity algorithm in \cite{Yu_07TOC_MACorder}
finds the near-optimal encoding/decoding order for \textbf{SPMP} by
observing that the optimal solution of \textbf{SPMP} must be the optimal
solution of some weighted sum-rate maximization problem, in which
the weight vector can be found and used to determine the decoding
order.

{}

\subsection{Contribution}

In summary, the \textbf{FOP} and \textbf{SPMP} for MIMO B-MAC networks
have been open problems. The contribution of the paper is as follows.
\begin{itemize}
\item \textit{Rate-SINR }\emph{Conversion}: One of the difficulties of solving
the problems is the joint optimization of beamforming matrices of
all links. One approach is to decompose a link to multiple single-input
single-output (SISO) streams and optimize the beamforming vectors
through SINR duality, if a bridge between rate and SINR can be built
to determine the optimal number of streams and rate/power allocation
among the streams. In Section \ref{sec:Rate-SINR}, we show that any
Pareto rate point of an achievable rate region can be mapped to a
Pareto SINR point of the achievable SINR region through two optimal
and simple mappings that produce equal rate and equal power streams
respectively. The significance of this result is that it offers a
method to convert the \textcolor{black}{rate problems to SINR problems.}
\item \textit{SINR based Algorithms}: Using the above result, we take advantage
of existing algorithms for SINR problems to solve\textcolor{black}{{}
}\textbf{\textcolor{black}{FOP}}\textcolor{black}{{} and }\textbf{\textcolor{black}{SPMP}}\textcolor{black}{{}
under rate constraints in Section \ref{sub:Algorithms-to-Solve} and
provide optimality analysis in Section \ref{sub:Optimality-Analysis}.}
\item \textcolor{black}{\emph{Polite Water-filling based Algorithms}}\textcolor{black}{:
}Another approach is to directly solve for the beamforming matrices.
For the convex problem of MIMO MAC, steepest ascent algorithm is used
except for the special case of sum-rate optimal points, where iterative
water-filling can be employed \cite{Yu_IT04_MIMO_MAC_waterfilling_alg,Jindal_IT05_IFBC,Weiyu_IT06_DualIWF}.
The B-MAC network problems are non-convex in general and thus, better
algorithms, like water-filling, than the steepest ascent algorithm
is highly desirable. However, directly applying traditional water-filling
is far from optimal \cite{Yu_JSAC02_Distributed_power_control_DSL,Popescu_Globecom03_Water_filling_not_good,Lai_IT08_water_filling_game_MAC}.
In \cite{Liu_IT10s_Duality_BMAC}, we recently found the long sought
optimal network version of water-filling, polite water-filling, which
is the optimal input structure of any Pareto rate point, not only
the sum-rate optimal point, of the achievable region of a MIMO B-MAC
network. This network version of water-filling is polite because it
optimally balances between reducing interference to others and maximizing
a link's own rate. The superiority of the polite water-filling is
demonstrated for weighted sum-rate maximization in \cite{Liu_IT10s_Duality_BMAC}
and the superiority is because it is hard not to obtain good results
when the optimal input structure is imposed to the solution at each
iteration. In Section \ref{sub:itree}, using polite water-filling,
we design an algorithm to monotonically improve the output of the
SINR based algorithms for iTree networks defined later, if the output
does not satisfy the KKT condition. Furthermore, in Section \ref{sub:PR-PR1},
purely polite water-filling based algorithms are designed that have
faster convergence speed.
\item \textit{Distributed Algorithm}: In a network, it is highly desirable
to use distributed algorithms. The polite water-filling based algorithm
is well suited for distributed implementation, which is shown in Section
\ref{sub:Distributed-Implementation}, where each node only needs
to estimate/exchange the local CSI but the performance of each iteration
is the same as that of the centralized algorithm.
\item \textit{Optimization of Encoding and Decoding Order}s: Another difficulty
is to find the optimal encoding/decoding order when interference cancellation
techniques like DPC/SIC are employed. Again, polite water-filling
proves useful in Section \ref{sub:OrderOptimization} because the
water-filling levels of the links can be used to identify the optimal
encoding/decoding order for BC/MAC and pseudo-BC/MAC defined later.
\end{itemize}

The rest of the paper is organized as follows. Section \ref{sec:System Model}
defines the achievable rate region and formulates the problems. Section
\ref{sec:Preliminary} summarizes the preliminaries on SINR duality
and polite water-filling. Section \ref{sec:Algorithms} presents the
efficient centralized and distributed algorithms. The performance
of the algorithms is verified by simulation in Section \ref{sec:Simulation-Results}.
The conclusion is given in Section \ref{sec:Conclusion}.

\section{\textcolor{black}{System Model and Problem Formulation\label{sec:System Model}}}

\subsection{Definition of the Achievable Rate Region}

\textcolor{black}{We consider a MIMO B-MAC interference network, consisting
of multiple interfering BCs and MACs. There are $L$ data links.}
Let $T_{l}$ and $R_{l}$ denote the virtual transmitter and receiver
of link $l$ equipped with $L_{T_{l}}$ transmit antennas and $L_{R_{l}}$
receive antennas respectively. The received signal at $R_{l}$ is

\begin{eqnarray}
\mathbf{y}_{l} & = & \sum_{k=1}^{L}\mathbf{H}_{l,k}\mathbf{x}_{k}+\mathbf{w}_{l},\label{eq:recvsignal}\end{eqnarray}
where $\mathbf{x}_{k}\in\mathbb{C}^{L_{T_{k}}\times1}$ is the transmit
signal of link $k$ and is assumed to be circularly symmetric complex
Gaussian; $\mathbf{H}_{l,k}\in\mathbb{C}^{L_{R_{l}}\times L_{T_{k}}}$
is the channel matrix between $T_{k}$ and $R_{l}$; and $\mathbf{w}_{l}\in\mathbb{C}^{L_{R_{l}}\times1}$
is a circularly symmetric complex Gaussian noise vector with zero
mean and identity covariance matrix.

To handle a wide range of interference cancellation possibilities,
we define a coupling matrix $\mathbf{\Phi}\in\mathbb{R}_{+}^{L\times L}$
as a function of the interference cancellation scheme \cite{Liu_IT10s_Duality_BMAC}.
It specifies whether interference is completely cancelled or treated
as noise: if $\mathbf{x}_{k}$, after interference cancellation, still
causes interference to $\mathbf{x}_{l}$, $\mathbf{\Phi}_{l,k}=1$
and otherwise, $\mathbf{\Phi}_{l,k}=0$. For example, if the virtual
transmitters (receivers) of several links are associated with the
same physical transmitter (receiver), interference cancellation techniques
such as dirty paper coding (successive decoding and cancellation)
can be applied at this physical transmitter (receiver) to improve
the performance.

The coupling matrices valid for the results of this paper are those
for which there exists a transmission and receiving scheme such that
each signal is decoded and possibly cancelled by no more than one
receiver. Possible extension to the Han-Kobayashi scheme, where a
common message is decoded by more than one receiver, is discussed
in \cite{Liu_IT10s_Duality_BMAC}. We give some examples of valid
coupling matrices. For a BC (MAC) employing DPC (SIC) where the $l^{\textrm{th}}$
link is the $l^{\textrm{th}}$ one to be encoded (decoded), the coupling
matrix is given by $\mathbf{\Phi}_{l,k}=0,\forall k\leq l$ and $\mathbf{\Phi}_{l,k}=1,\forall k>l$.
In Fig. \ref{fig:sysFig1}, we give an example of a B-MAC network
employing DPC and SIC. When no data is transmitted over link 4 and
5, the following $\mathbf{\Phi}^{a},\mathbf{\Phi}^{b},\mathbf{\Phi}^{c},\mathbf{\Phi}^{d}$
are valid coupling matrices for link $1,2,3$ under the corresponding
encoding and decoding orders: \emph{a}. $\mathbf{x}_{1}$ is encoded
after $\mathbf{x}_{2}$ and $\mathbf{x}_{2}$ is decoded after $\mathbf{x}_{3}$;
\emph{b}. $\mathbf{x}_{2}$ is encoded after $\mathbf{x}_{1}$ and
$\mathbf{x}_{2}$ is decoded after $\mathbf{x}_{3}$; \emph{c}. $\mathbf{x}_{1}$
is encoded after $\mathbf{x}_{2}$ and $\mathbf{x}_{3}$ is decoded
after $\mathbf{x}_{2}$; \emph{d}. There is no interference cancellation.\begin{align*}
\mathbf{\Phi}^{a}=\left[\begin{array}{ccc}
0 & 0 & 1\\
1 & 0 & 0\\
1 & 1 & 0\end{array}\right], & \:\mathbf{\Phi}^{b}=\left[\begin{array}{ccc}
0 & 1 & 1\\
0 & 0 & 0\\
1 & 1 & 0\end{array}\right],\\
\mathbf{\Phi}^{c}=\left[\begin{array}{ccc}
0 & 0 & 1\\
1 & 0 & 1\\
1 & 0 & 0\end{array}\right], & \:\mathbf{\Phi}^{d}=\left[\begin{array}{ccc}
0 & 1 & 1\\
1 & 0 & 1\\
1 & 1 & 0\end{array}\right].\end{align*}

\begin{figure}
\begin{centering}
\textsf{\includegraphics[clip,scale=0.3]{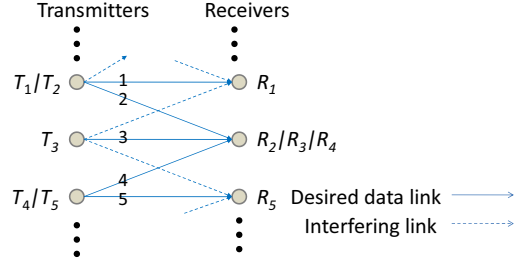}}
\par\end{centering}

\caption{\label{fig:sysFig1}Example of a B-MAC network}

\end{figure}

Note that when DPC and SIC are combined, an interference may not be
fully cancelled under a specific encoding and decoding order. Such
case cannot be described by the coupling matrix of 0's and 1's defined
above. But a valid coupling matrix can serve for an upper or lower
bound. See more discussion in \cite{Liu_IT10s_Duality_BMAC}.

If not explicitly stated otherwise, achievable regions in this paper
refer to the following. Note that $\mathbf{\Phi}_{l,l}=0$ by definition.
The interference-plus-noise covariance matrix of the $l^{\text{th}}$
link is \begin{eqnarray}
\mathbf{\Omega}_{l} & = & \mathbf{I}+\sum_{k=1}^{L}\mathbf{\Phi}_{l,k}\mathbf{H}_{l,k}\mathbf{\Sigma}_{k}\mathbf{H}_{l,k}^{\dagger},\label{eq:whiteMG}\end{eqnarray}
where $\mathbf{\Sigma}_{k}$ is the covariance matrix of $\mathbf{x}_{k}$.
We denote all the covariance matrices as \begin{eqnarray}
\mathbf{\Sigma}_{1:L} & = & \left(\mathbf{\Sigma}_{1},\mathbf{\Sigma}_{2},...,\mathbf{\Sigma}_{L}\right).\label{eq:TQdefG}\end{eqnarray}
Then the achievable mutual information (rate) of link $l$ is given
by a function of $\mathbf{\Sigma}_{1:L}$ and $\mathbf{\Phi}$\begin{eqnarray}
\mathcal{I}_{l}\left(\mathbf{\Sigma}_{1:L},\mathbf{\Phi}\right) & = & \textrm{log}\left|\mathbf{I}+\mathbf{H}_{l,l}\mathbf{\Sigma}_{l}\mathbf{H}_{l,l}^{\dagger}\mathbf{\Omega}_{l}^{-1}\right|.\label{eq:linkkMIG}\end{eqnarray}

\begin{definitn}
The\emph{ Achievable Rate Region }with \textit{\emph{a fixed coupling
matrix $\mathbf{\Phi}$}} and sum power constraint $P_{T}$ is defined
as
\end{definitn}
\begin{eqnarray}
\mathcal{R}_{\mathbf{\Phi}}\left(P_{T}\right) & \triangleq & \underset{\mathbf{\Sigma}_{1:L}:\sum_{l=1}^{L}\textrm{Tr}\left(\mathbf{\mathbf{\Sigma}}_{l}\right)\leq P_{T}}{\bigcup}\left\{ \mathbf{r}\in\mathbb{R}_{+}^{L}:\right.\label{eq:Ratereg1G}\\
 &  & \left.r_{l}\le\mathcal{I}_{l}\left(\mathbf{\Sigma}_{1:L},\mathbf{\Phi}\right),1\leq l\leq L\right\} .\nonumber \end{eqnarray}

A bigger achievable rate region can be defined by the convex closure
of $\bigcup_{\mathbf{\Phi}\in\mathbf{\Xi}}\mathcal{R}_{\mathbf{\Phi}}\left(P_{T}\right)$,
where $\mathbf{\Xi}$ is a set of valid coupling matrices. For example,
if DPC and/or SIC are employed, $\mathbf{\Xi}$ can be a set of valid
coupling matrices corresponding to various valid encoding and/or decoding
orders.

The algorithms rely on the duality between the forward and reverse
links of a B-MAC network. The reverse links are obtained by reversing
the transmission direction and replacing the channel matrices by their
conjugate transposes. \textcolor{black}{The coupling matrix for the
reverse links is the transpose of that for the forward links}. We
use the notation $\hat{}$ to denote the corresponding terms in the
reverse links. For example, in the reverse links of the B-MAC network
in Fig. \ref{fig:sysFig1}, $T_{1}/T_{2}$ ($R_{2}/R_{3}$) becomes
the receiver (transmitter), and $\mathbf{\hat{x}}_{2}$ is decoded
after $\hat{\mathbf{x}}_{1}$ and $\mathbf{\hat{x}}_{3}$ is encoded
after $\hat{\mathbf{x}}_{2}$, if in the forward links, $\mathbf{x}_{1}$
is encoded after $\mathbf{x}_{2}$ and $\mathbf{x}_{2}$ is decoded
after $\mathbf{x}_{3}$. The interference-plus-noise covariance matrix
of reverse link $l$ is \begin{eqnarray}
\hat{\mathbf{\Omega}}_{l} & = & \mathbf{I}+\sum_{k=1}^{L}\mathbf{\Phi}_{k,l}\mathbf{H}_{k,l}^{\dagger}\mathbf{\hat{\mathbf{\Sigma}}}_{k}\mathbf{H}_{k,l},\label{eq:WhiteMRV}\end{eqnarray}
and the rate of reverse link $l$ is given by $\mathcal{\hat{I}}_{l}\left(\hat{\mathbf{\Sigma}}_{1:L},\mathbf{\Phi}^{T}\right)=\textrm{log}\left|\mathbf{I}+\mathbf{H}_{l,l}^{\dagger}\hat{\mathbf{\Sigma}}_{l}\mathbf{H}_{l,l}\hat{\mathbf{\Omega}}_{l}^{-1}\right|.$

\subsection{Problem Formulation}

\textcolor{black}{This paper concerns the feasibility optimization
problem }\textbf{(FOP)}\textcolor{black}{{} and the sum power minimization
problem (}\textbf{\textcolor{black}{SPMP}}\textcolor{black}{) under
Quality of Service (QoS) constraints in terms of target rates $\left[\mathcal{I}_{l}^{0}\right]_{l=1,...,L}$
for a B-MAC network with a given valid} coupling matrix $\mathbf{\Phi}$\textcolor{black}{:\begin{eqnarray}
\textrm{\textbf{FOP}:} & \underset{\mathbf{\Sigma}_{1:L}}{\textrm{max}} & \left(\min_{1\leq l\leq L}\frac{\mathcal{I}_{l}\left(\mathbf{\Sigma}_{1:L},\mathbf{\Phi}\right)}{\mathcal{I}_{l}^{0}}\right)\label{eq:P1}\\
 & \textrm{s.t.} & \mathbf{\Sigma}_{l}\succeq0,l=1,...,L\:\textrm{and}\:\sum_{l=1}^{L}\textrm{Tr}\left(\mathbf{\Sigma}_{l}\right)\leq P_{T},\nonumber \end{eqnarray}
where $P_{T}$ is the total power constraint; \begin{eqnarray}
\textrm{\textbf{SPMP}:} & \underset{\mathbf{\Sigma}_{1:L}}{\textrm{min}} & \sum_{l=1}^{L}\textrm{Tr}\left(\mathbf{\Sigma}_{l}\right)\label{eq:P2}\\
 & \textrm{s.t.} & \mathcal{I}_{l}\left(\mathbf{\Sigma}_{1:L},\mathbf{\Phi}\right)\geq\mathcal{I}_{l}^{0},\mathbf{\Sigma}_{l}\succeq0,l=1,...,L.\nonumber \end{eqnarray}
In }\textbf{\textcolor{black}{FOP}}\textcolor{black}{, if the optimum
of the objective function is $\alpha$ and is greater than one, the
target rates are feasible. The optimum input covariance matrices achieves
a point on the boundary of the achievable region along the direction
of vector $\left[\mathcal{I}_{l}^{0}\right]_{l=1,...,L}$, i.e., the
optimal rate vector satisfies $\left[\mathcal{I}_{l}\right]_{l=1,...,L}=\alpha\left[\mathcal{I}_{l}^{0}\right]_{l=1,...,L}$.
If the target rates is feasible for some power, }\textbf{\textcolor{black}{SPMP}}\textcolor{black}{{}
finds the minimum total power needed.} For the special case of DPC
and SIC, the optimal \textit{\emph{coupling matrix}} $\mathbf{\Phi}$,
or equivalently, the optimal encoding and/or decoding order of \textbf{\textcolor{black}{FOP}}\textcolor{black}{{}
and }\textbf{\textcolor{black}{SPMP}} is partially solved in Section\textcolor{black}{{}
\ref{sub:OrderOptimization}. We first focus on centralized algorithms
under total power constraints. Then we give a distributed implementation
of the algorithm for }\textbf{\textcolor{black}{SPMP}}\textcolor{black}{{}
under additional individual maximum power constraints.}

Although we focus on the sum power and white noise in this paper for
simplicity, the results can be directly applied to a much larger class
of problems with a single linear constraint $\sum_{l=1}^{L}\textrm{Tr}\left(\mathbf{\Sigma}_{l}\hat{\mathbf{W}}_{l}\right)\leq P_{T}$
in \textbf{FOP} (or objective function $\sum_{l=1}^{L}\textrm{Tr}\left(\mathbf{\Sigma}_{l}\hat{\mathbf{W}}_{l}\right)$
in \textbf{SPMP}) and/or colored noise with covariance $\text{E}\left[\mathbf{w}_{l}\mathbf{\mathbf{w}}_{l}^{\dagger}\right]=\mathbf{W}_{l}$,
which includes the weighted sum power minimization problem in \cite{Rao_TOC07_netduality}
as a special case. Only variable changes $\mathbf{\Sigma}_{l}'=\hat{\mathbf{W}}_{l}^{\frac{1}{2}}\mathbf{\Sigma}_{l}\hat{\mathbf{W}}_{l}^{\frac{1}{2}}$
and $\mathbf{W}_{k}^{-\frac{1}{2}}\mathbf{H}_{k,l}\hat{\mathbf{W}}_{l}^{-\frac{1}{2}}$
are needed, where $\hat{\mathbf{W}}_{l}$ and $\mathbf{W}_{k}$ are
positive definite for meaningful cases%
\footnote{For random channels, singular $\hat{\mathbf{W}}_{l}$ or $\mathbf{W}_{l}$
will result in infinite power and rate with probability one.%
}. The single linear constraint appears in Lagrange functions for problems
with multiple linear constraints \cite{Yu_IT06_Minimax_duality,Zhang_IT08_MACBC_LC},
and thus, the results in this paper serve as the basis to solve them
\cite{Liu_10Allerton_MLC}. Special cases of multiple linear constraints
include individual power constraints, per-antenna power constraints,
interference constraints in cognitive radios, etc..

\section{\textcolor{black}{Preliminaries\label{sec:Preliminary}}}

The algorithms are based on SINR duality, e.g., \cite{Rao_TOC07_netduality},
rate duality, and polite water-filling developed earlier \cite{Liu_IT10s_Duality_BMAC}.
They are reviewed below.

\subsection{\textcolor{black}{SINR Duality for MIMO B-MAC Networks\label{sub:precoding}}}

\textcolor{black}{The achievable rate region defined in (\ref{eq:Ratereg1G})
can be achieved by a spatial multiplexing scheme as follows. }
\begin{definitn}
\textcolor{black}{\label{def:Decomposition}The }\textcolor{black}{\emph{Decomposition
of a MIMO Link into Multiple SISO Data Streams}}\textcolor{black}{{}
is defined as, for link $l$ and $M_{l}\ge\text{Rank}(\mathbf{\Sigma}_{l})$,
finding a precoding matrix $\dot{\mathbf{T}}_{l}=\left[\sqrt{p_{l,1}}\mathbf{t}_{l,1},...,\sqrt{p_{l,M_{l}}}\mathbf{t}_{l,M_{l}}\right]$
satisfying \begin{equation}
\mathbf{\Sigma}_{l}=\dot{\mathbf{T}}_{l}\dot{\mathbf{T}}_{l}^{\dagger}=\sum_{m=1}^{M_{l}}p_{l,m}\mathbf{t}_{l,m}\mathbf{t}_{l,m}^{\dagger},\label{eq:DecomSig}\end{equation}
where $\mathbf{t}_{l,m}\in\mathbb{C}^{L_{T_{l}}\times1}$ is a transmit
vector with $\left\Vert \mathbf{t}_{l,m}\right\Vert =1$; and $\mathbf{p}=\left[p_{1,1},...,p_{1,M_{1}},...,p_{L,1},...,p_{L,M_{L}}\right]^{T}$
are the transmit powers. }
\end{definitn}

\textcolor{black}{Note that the precoding matrix is not unique because
$\dot{\mathbf{T}}_{l}^{'}=\dot{\mathbf{T}}_{l}\mathbf{V}$ with unitary
$\mathbf{V}\in\mathbb{C}^{M_{l}\times M_{l}}$ also gives the same
covariance matrix in (\ref{eq:DecomSig}). Without loss of generality,
we assume the intra-signal decoding order is that the $m^{\text{th}}$
stream is the $m^{\text{th}}$ to be decoded and cancelled. The receive
vector for the $m^{th}$ stream of link $l$ is obtained by the MMSE
filtering as \begin{equation}
\mathbf{r}_{l,m}=\alpha_{l,m}\left(\sum_{i=m+1}^{M_{l}}\mathbf{H}_{l,l}p_{l,i}\mathbf{t}_{l,i}\mathbf{t}_{l,i}^{\dagger}\mathbf{H}_{l,l}^{\dagger}+\mathbf{\Omega}_{l}\right)^{-1}\mathbf{H}_{l,l}\mathbf{t}_{l,m},\label{eq:MMSErev1G}\end{equation}
where $\alpha_{l,m}$ is chosen such that $\left\Vert \mathbf{r}_{l,m}\right\Vert =1$.
This is referred to as MMSE-SIC receiver in this paper. }

\textcolor{black}{For each stream, one can calculate its SINR. Let
the collections of transmit and receive vectors be \begin{eqnarray}
\mathbf{T} & = & \left[\mathbf{t}_{l,m}\right]_{m=1,...,M_{l},l=1,...,L},\label{eq:udef}\\
\mathbf{R} & = & \left[\mathbf{r}_{l,m}\right]_{m=1,...,M_{l},l=1,...,L}.\label{eq:TRkdef}\end{eqnarray}
The cross-talk matrix $\mathbf{\Psi}\left(\mathbf{T},\mathbf{R}\right)\in\mathbb{R}_{+}^{\sum_{l}M_{l}\times\sum_{l}M_{l}}$
between different streams \cite{Martin_ITV_04_BFdual} is a function
of $\mathbf{T},\mathbf{R}$, and, assuming unit transmit power, the
element of the $\left(\sum_{i=1}^{l-1}M_{i}+m\right)^{\textrm{th}}$
row and $\left(\sum_{i=1}^{k-1}M_{i}+n\right)^{\textrm{th}}$ column
of $\mathbf{\Psi}$ is the interference power from the $k^{\text{th}}$
link's $n^{\text{th}}$ stream to the $l^{\text{th}}$ link's $m^{\text{th}}$
stream and is given by \begin{align}
\mathbf{\Psi}_{l,m}^{k,n}= & \begin{cases}
0 & k=l\:\textrm{and}\: m\geq n,\\
\left|\mathbf{r}_{l,m}^{\dagger}\mathbf{H}_{l,l}\mathbf{t}_{l,n}\right|^{2} & k=l,\:\textrm{and}\: m<n,\\
\Phi_{l,k}\left|\mathbf{r}_{l,m}^{\dagger}\mathbf{H}_{l,k}\mathbf{t}_{k,n}\right|^{2} & \textrm{otherwise}.\end{cases}\label{eq:faiG}\end{align}
Then the SINR for the $m^{th}$ stream of link $l$ is \begin{eqnarray}
\gamma_{l,m}\left(\mathbf{T},\mathbf{R},\mathbf{p}\right) & = & \frac{p_{l,m}\left|\mathbf{r}_{l,m}^{\dagger}\mathbf{H}_{l,l}\mathbf{t}_{l,m}\right|^{2}}{1+{\displaystyle \sum_{k=1}^{L}}{\displaystyle \sum_{n=1}^{M_{k}}}p_{k,n}\mathbf{\Psi}_{l,m}^{k,n}}.\label{eq:SINR1G}\end{eqnarray}
Such decomposition of data to streams with MMSE-SIC receiver is information
lossless \cite{Varanasi_Asilomar97_MMSE_is_optimal}, i.e., the sum-rate
of all streams of link $l$ is equal to the mutual information in
(\ref{eq:linkkMIG}).}

\textcolor{black}{In the reverse links, we can obtain SINRs using
$\mathbf{R}$ as transmit vectors and $\mathbf{T}$ as receive vectors.
The transmit powers are denoted as $\mathbf{q}=\left[q_{1,1},...,q_{1,M_{1}},...,q_{L,1},...,q_{L,M_{L}}\right]^{T}$.
The intra-signal decoding order is the opposite to that of the forward
link, i.e., the $m^{\text{th}}$ stream is the $m^{\text{th}}$ last
to be decoded and cancelled. Then the SINR for the $m^{th}$ stream
of reverse link $l$ is\begin{eqnarray}
\hat{\gamma}_{l,m}\left(\mathbf{R},\mathbf{T},\mathbf{q}\right) & = & \frac{q_{l,m}\left|\mathbf{t}_{l,m}^{\dagger}\mathbf{H}_{l,l}^{\dagger}\mathbf{r}_{l,m}\right|^{2}}{1+{\displaystyle \sum_{k=1}^{L}}{\displaystyle \sum_{n=1}^{M_{k}}}q_{k,n}\mathbf{\Psi}_{k,n}^{l,m}}.\label{eq:SINR2G}\end{eqnarray}
For simplicity, we will use $\left\{ \mathbf{T},\mathbf{R},\mathbf{p}\right\} $
($\left\{ \mathbf{R},\mathbf{T},\mathbf{q}\right\} $) to denote the
transmission and reception strategy described above in the forward
(reverse) links. }

\textcolor{black}{The achievable SINR regions of the forward and reverse
links are the same. Define the achievable SINR regions }$\mathcal{T}_{\mathbf{\Phi}}\left(P_{T}\right)$
and $\mathcal{\hat{T}}_{\mathbf{\Phi}^{T}}\left(P_{T}\right)$\textcolor{black}{{}
as the set of all SINRs that can be achieved under the sum power constraint
$P_{T}$ in the forward and reverse links respectively. For a given
set of SINR values $\mathbf{\gamma}^{0}=\left[\gamma_{l,m}^{0}\right]_{m=1,...,M_{l},l=1,...,L}$,
define a diagonal matrix $\mathbf{D}\left(\mathbf{T},\mathbf{R},\mathbf{\gamma}^{0}\right)\in\mathbb{R}_{+}^{\sum_{l}M_{l}\times\sum_{l}M_{l}}$
where the $\left(\sum_{i=1}^{l-1}M_{i}+m\right)^{\text{th}}$ diagonal
element is }

\textcolor{black}{\begin{eqnarray}
\mathbf{D}_{\sum_{i=1}^{l-1}M_{i}+m,\sum_{i=1}^{l-1}M_{i}+m} & = & \gamma_{l,m}^{0}/\left|\mathbf{r}_{l,m}^{\dagger}\mathbf{H}_{l,l}\mathbf{t}_{l,m}\right|^{2}.\label{eq:DG}\end{eqnarray}
We restate the SINR duality, e.g. \cite{Rao_TOC07_netduality}, as
follows.}
\begin{lemma}
\textcolor{black}{\label{lem:lem1G}If a set of SINRs $\mathbf{\gamma}^{0}$
is achieved by the transmission and reception strategy $\left\{ \mathbf{T},\mathbf{R},\mathbf{p}\right\} $
with $\left\Vert \mathbf{p}\right\Vert _{1}=P_{T}$ in the forward
links, then $\mathbf{\gamma}^{0}$ is also achievable in the reverse
links with $\left\{ \mathbf{R},\mathbf{T},\mathbf{q}\right\} $, where
$\mathbf{q}$ satisfies $\left\Vert \mathbf{q}\right\Vert _{1}=P_{T}$
and is given by\begin{eqnarray}
\mathbf{q} & = & \left(\mathbf{D}^{-1}\left(\mathbf{T},\mathbf{R},\mathbf{\gamma}^{0}\right)-\mathbf{\Psi}^{T}\left(\mathbf{T},\mathbf{R}\right)\right)^{-1}\mathbf{1}.\label{eq:qpower}\end{eqnarray}
And thus, one has }$\mathcal{T}_{\mathbf{\Phi}}\left(P_{T}\right)=\mathcal{\hat{T}}_{\mathbf{\Phi}^{T}}\left(P_{T}\right)$.
\end{lemma}

\subsection{\textcolor{black}{Rate Duality\label{sub:Main-Results}}}

\textcolor{black}{The rate duality of the forward and reverse links
of the B-MAC networks is a simple consequence of the SINR duality
\cite{Liu_IT10s_Duality_BMAC}}. The reverse link input covariance
matrices are obtained by the following transformation.

\begin{definitn}
Let $\mathbf{\Sigma}_{l}=\sum_{m=1}^{M_{l}}p_{l,m}\mathbf{t}_{l,m}\mathbf{t}_{l,m}^{\dagger},l=1,...,L$
be a decomposition of $\mathbf{\Sigma}_{1:L}$. Compute the MMSE-SIC
receive vectors $\mathbf{R}$ from (\ref{eq:MMSErev1G}) and the reverse
transmit powers $\mathbf{q}$ from (\ref{eq:qpower}). The \emph{Covariance
Transformation} from $\mathbf{\Sigma}_{1:L}$ to $\hat{\mathbf{\Sigma}}_{1:L}$
is

\begin{eqnarray}
\hat{\mathbf{\Sigma}}_{l} & = & \sum_{m=1}^{M_{l}}q_{l,m}\mathbf{r}_{l,m}\mathbf{r}_{l,m}^{\dagger},l=1,...,L.\label{eq:CovTrans}\end{eqnarray}

\end{definitn}

We give the rate duality under a linear constraint $\sum_{l=1}^{L}\textrm{Tr}\left(\mathbf{\Sigma}_{l}\hat{\mathbf{W}}_{l}\right)\leq P_{T}$
and/or colored noise with covariance $\text{E}\left[\mathbf{w}_{l}\mathbf{\mathbf{w}}_{l}^{\dagger}\right]=\mathbf{W}_{l}$.
The covariance transformation for this case is also calculated from
the MMSE receive beams and power allocation that makes SINRs of the
forward and reverse links equal\textcolor{black}{, as in (\ref{eq:CovTrans}).
The only difference }is that the identity noise covariance in $\mathbf{\Omega}_{l}$
is replaced by $\mathbf{W}_{l}$ and the all-one vector $\mathbf{1}$
in (\ref{eq:qpower}) is replaced by the vector $\left[\mathbf{t}_{l,m}^{\dagger}\hat{\mathbf{W}}_{l}\mathbf{t}_{l,m}\right]_{m=1,...,M_{l},l=1,...,L}$.%
{} For convenience, let \begin{equation}
\left(\left[\mathbf{H}_{l,k}\right],\sum_{l=1}^{L}\textrm{Tr}\left(\mathbf{\Sigma}_{l}\hat{\mathbf{W}}_{l}\right)\leq P_{T},\left[\mathbf{W}_{l}\right]\right),\label{eq:net-color-linear-constraint}\end{equation}
denote a network where the channel matrices are $\left[\mathbf{H}_{l,k}\right]$;
the input covariance matrices must satisfy the linear constraint $\sum_{l=1}^{L}\textrm{Tr}\left(\mathbf{\Sigma}_{l}\hat{\mathbf{W}}_{l}\right)\leq P_{T}$;
and the covariance matrix of the noise at the receiver of link $l$
is $\mathbf{W}_{l}$. \textcolor{black}{Then the rate duality is restated
in the theorem below.}
\begin{thm}
\label{thm:linear-color-dual}The dual of the network (\ref{eq:net-color-linear-constraint})
is\begin{equation}
\left(\left[\mathbf{H}_{k,l}^{\dagger}\right],\sum_{l=1}^{L}\textrm{Tr}\left(\hat{\mathbf{\Sigma}}_{l}\mathbf{W}_{l}\right)\leq P_{T},\left[\hat{\mathbf{W}}_{l}\right]\right)\label{eq:net-forward-color-dual}\end{equation}
in the sense that 1) they have the same achievable rate region; 2)
if $\mathbf{\Sigma}_{1:L}$ achieves certain rates and satisfies the
linear constraint in network (\ref{eq:net-color-linear-constraint}),
its covariance transformation $\hat{\mathbf{\Sigma}}_{1:L}$ achieves
better rates in network (\ref{eq:net-forward-color-dual}) under the
linear constraint $\sum_{l=1}^{L}\textrm{Tr}\left(\hat{\mathbf{\Sigma}}_{l}\mathbf{W}_{l}\right)=\sum_{l=1}^{L}\textrm{Tr}\left(\mathbf{\Sigma}_{l}\hat{\mathbf{W}}_{l}\right)\leq P_{T}$.
\end{thm}

\subsection{Polite Water-filling}

\textcolor{black}{In \cite{Liu_IT10s_Duality_BMAC}, we showed that
the Pareto optimal input covariance matrices have a }\textit{\textcolor{black}{polite
water-filling structure}}\textcolor{black}{, which is defined below.
It generalizes the well known optimal single user water-filling structure
to networks.}
\begin{definitn}
\textcolor{black}{}%
{}\textcolor{black}{\label{def:DefGWF}Given input covariance matrices
$\mathbf{\Sigma}_{1:L}$, obtain }its covariance transformation $\hat{\mathbf{\Sigma}}_{1:L}$
as \textcolor{black}{in (\ref{eq:CovTrans})}. \textcolor{black}{Let
$\mathbf{\Omega}_{l}$'s and }$\hat{\mathbf{\Omega}}_{l}$'s\textcolor{black}{{}
respectively be the corresponding interference-plus-noise covariance
matrices. For each link $l$, pre- and post- whiten the channel $\mathbf{H}_{l,l}$
to produce an equivalent single user channel }$\bar{\mathbf{H}}_{l}=\mathbf{\Omega}_{l}^{-1/2}\mathbf{H}_{l,l}\hat{\mathbf{\Omega}}_{l}^{-1/2}$\textcolor{black}{.
Define }$\mathbf{Q}_{l}\triangleq\hat{\mathbf{\Omega}}_{l}^{1/2}\mathbf{\Sigma}_{l}\hat{\mathbf{\Omega}}_{l}^{1/2}$
\textcolor{black}{as the equivalent input covariance matrix of link
$l$. }The input covariance matrix $\mathbf{\Sigma}_{l}$ is said
to \textcolor{black}{possess a }\textcolor{black}{\emph{polite water-filling
structure}}\textcolor{black}{{} if} $\mathbf{Q}_{l}$ satisfies the
structure of water-filling over $\bar{\mathbf{H}}_{l}$, i.e.,\textcolor{black}{{}
\begin{eqnarray}
\mathbf{Q}_{l} & = & \mathbf{G}_{l}\mathbf{D}_{l}\mathbf{G}_{l}^{\dagger},\label{eq:WFFar}\\
\mathbf{D}_{l} & = & \left(\nu_{l}\mathbf{I}-\mathbf{\Delta}_{l}^{-2}\right)^{+}.\nonumber \end{eqnarray}
where $\nu_{l}\geq0$ is called the }\textit{\textcolor{black}{polite
water-filling level}}\textcolor{black}{; the equivalent channel $\bar{\mathbf{H}}_{l}$'s
thin singular value decomposition (SVD) is $\bar{\mathbf{H}}_{l}=\mathbf{F}_{l}\mathbf{\Delta}_{l}\mathbf{G}_{l}^{\dagger}$
with $\mathbf{F}_{l}\in\mathbb{C}^{L_{R_{l}}\times N_{l}},\ \mathbf{G}_{l}\in\mathbb{C}^{L_{T_{l}}\times N_{l}},\ \mathbf{\Delta}_{l}\in\mathbb{R}_{++}^{N_{l}\times N_{l}}$,
and $N_{l}=\textrm{Rank}\left(\mathbf{H}_{l,l}\right)$. If all }$\mathbf{\Sigma}_{l}$'s
\textcolor{black}{possess the polite water-filling structure, then
$\mathbf{\Sigma}_{1:L}$ is said to possess the polite water-filling
structure.}\end{definitn}
\begin{thm}
\label{thm:WFST}The input covariance matrices $\mathbf{\Sigma}_{1:L}$
of a Pareto rate point of the achievable region and its covariance
transformation $\hat{\mathbf{\Sigma}}_{1:L}$ \textcolor{black}{possess}
the polite water-filling structure.
\end{thm}

The following theorem proved in \textcolor{black}{\cite{Liu_IT10s_Duality_BMAC}}
states that $\mathbf{\Sigma}_{l}$ having the polite water-filling
structure suffices for $\hat{\mathbf{\Sigma}}_{l}$ to have the polite
water-filling structure even at a non-Pareto rate point.
\begin{thm}
\label{thm:FequRGWF}If one input covariance matrix $\mathbf{\Sigma}_{l}$
has the polite water-filling structure while other $\mathbf{\Sigma}_{k},\hat{\mathbf{\Sigma}}_{k},$
$k\ne l,$ are fixed, so does its covariance transformation $\hat{\mathbf{\Sigma}}_{l}$,
i.e., $\mathbf{\hat{Q}}_{l}\triangleq\mathbf{\Omega}_{l}^{1/2}\hat{\mathbf{\Sigma}}_{l}\mathbf{\Omega}_{l}^{1/2}$
satisfies the structure of water-filling over the reverse equivalent
channel $\bar{\mathbf{H}}_{l}^{\dagger}\triangleq\hat{\mathbf{\Omega}}_{l}^{-1/2}\mathbf{H}_{l,l}^{\dagger}\mathbf{\Omega}_{l}^{-1/2}$.
Further more, $\hat{\mathbf{\Sigma}}_{l}$ can be expressed as \begin{align}
\hat{\mathbf{\Sigma}}_{l} & =\nu_{l}\left(\mathbf{\Omega}_{l}^{-1}-\left(\mathbf{H}_{l,l}\mathbf{\Sigma}_{l}\mathbf{H}_{l,l}^{\dagger}+\mathbf{\Omega}_{l}\right)^{-1}\right),\ l=1,...,L,\label{eq:SigmhDirect}\end{align}
where $\nu_{l}$ is the polite water-filling level in (\ref{eq:WFFar}).
\end{thm}

\section{\label{sec:Algorithms}Optimization Algorithms}

\textcolor{black}{In this section, we present several related algorithms
for the feasibility optimization problem }\textbf{(FOP)}\textcolor{black}{{}
and the sum power minimization problem (}\textbf{\textcolor{black}{SPMP}}\textcolor{black}{)
under rate constraints. SINR based and polite water-filling based
algorithms are designed. Algorithms for SINR version of }\textbf{\textcolor{black}{FOP}}\textcolor{black}{{}
and }\textbf{\textcolor{black}{SPMP}}\textcolor{black}{{} have been
designed in \cite{Martin_ITV_04_BFdual,Rao_TOC07_netduality}. To
take advantage of them, we show how to map a Pareto point of the achievable
rate region to a Pareto point of the SINR region in Section \ref{sec:Rate-SINR}
and then use SINR based Algorithm A and B to solve }\textbf{\textcolor{black}{FOP}}\textcolor{black}{{}
and }\textbf{\textcolor{black}{SPMP}}\textcolor{black}{{} respectively
in Section \ref{sub:Algorithms-to-Solve}. The optimality of Algorithms
A and B is studied in Section \ref{sub:Optimality-Analysis} by examining
the structure of the optimal solutions of }\textbf{\textcolor{black}{FOP}}\textcolor{black}{{}
and }\textbf{\textcolor{black}{SPMP}}\textcolor{black}{. Then, for
iTree networks defined later, Algorithm I is designed to improve the
output of Algorithm A and B. The improvement and the optimal structure
suggests that the rate constrained problems can be directly solved
using Algorithm PR and PR1 in Section \ref{sub:PR-PR1} by polite
water-filling, without resorting to the SINR based approach. In a
network, it is desirable to have distributed algorithms, for which
Algorithm PRD is designed in Section \ref{sub:Distributed-Implementation}.
Finally, we design Algorithm O to improve the encoding and decoding
orders for all of the above algorithms when DPC and SIC are employed.
For convenience, a list of algorithms in this paper is summarized
in Table \ref{tab:List}.}

\textit{\textcolor{black}{\emph{}}}%
{}

\begin{table}
\caption{\label{tab:List}List of Algorithms}

\centering{}{\small }\begin{tabular}{|c|c|c|c|}
\hline
{\small Sec.} & {\small Tab.} & {\small Alg.} & {\small Purpose}\tabularnewline
\hline
\hline
{\small \ref{sub:Algorithms-to-Solve}} & {\small \ref{tab:table1}} & {\small A} & {\small SINR based, for }\textbf{\small EFOP}\tabularnewline
\hline
{\small \ref{sub:Algorithms-to-Solve}} & {\small \ref{tab:table2}} & {\small B} & {\small SINR based, for }\textbf{\small ESPMP}\tabularnewline
\hline
{\small \ref{sub:itree}} & {\small \ref{tab:table4}} & {\small S} & {\small Subroutine, for link $i$}\tabularnewline
\hline
{\small \ref{sub:itree}} & {\small \ref{tab:table3}} & {\small I} & {\small Improvement of A/B for iTree networks}\tabularnewline
\hline
{\small \ref{sub:PR-PR1}} & {\small \ref{tab:alg-W}} & {\small W} & {\small Subroutine, for water-filling level}\tabularnewline
\hline
{\small \ref{sub:PR-PR1}} & {\small \ref{tab:table7}} & {\small PR} & {\small Polite WF based, for iTree, }\textbf{\small SPMP}\tabularnewline
\hline
{\small \ref{sub:PR-PR1}} & {\small \ref{tab:table6}} & {\small PR1} & {\small Polite WF based, for B-MAC, }\textbf{\small SPMP}\tabularnewline
\hline
{\small \ref{sub:Distributed-Implementation}} & {\small \ref{tab:table9}} & {\small PRD} & {\small Distributed version of PR1}\tabularnewline
\hline
{\small \ref{sub:OrderOptimization}} & {\small \ref{tab:table8}} & {\small O} & {\small Enc./dec. order optimization, for }\textbf{\small FOP}{\small /}\textbf{\small SPMP}\tabularnewline
\hline
\end{tabular}
\end{table}

\subsection{Rate-SINR Conversion\label{sec:Rate-SINR}}

In order to find Pareto rate points of the achievable rate region
by taking advantage of algorithms that finds Pareto points of the
SINR region, one needs to find a mapping from a Pareto rate point
to a Pareto SINR point. But multiple SINR points can correspond to
the same rate and thus, multiple mappings exist. The following two
theorems give an equal SINR mapping and an equal power mapping by\textcolor{black}{{}
}choosing two decompositions of a MIMO link to multiple SISO data
streams. Note that for the same total link rate, different decompositions
have different sets of SINRs of the streams and different number of
streams. \textcolor{black}{We show that equal SINR allocation or equal
power allocation among the streams within a link will not lose optimality.}
\begin{thm}
\textcolor{black}{\label{thm:EquSINRopt}For any input covariance
matrices $\mathbf{\Sigma}_{1:L}$ achieving a rate point $\left[\mathcal{I}_{l}\right]_{l=1,...,L}$,
there exists a decomposition $\mathbf{\Sigma}_{l}=\dot{\mathbf{T}}_{l}\dot{\mathbf{T}}_{l}^{\dagger}=\sum_{m=1}^{M_{l}}p_{l,m}\mathbf{t}_{l,m}\mathbf{t}_{l,m}^{\dagger},l=1,...,L$,
with $M_{l}\ge\text{Rank}(\mathbf{\Sigma}_{l})$, such that the corresponding
transmission and MMSE-SIC reception strategy $\left\{ \mathbf{T},\mathbf{R},\mathbf{p}\right\} $
achieves equal SINR for all streams of the same link, i.e., $\gamma_{l,m}=e^{\mathcal{I}_{l}/M_{l}}-1,m=1,...,M_{l},l=1,...,L$.
Therefore uniform rate allocation over the streams of the same link
will not lose optimality.}
\end{thm}

\textcolor{black}{The proof is given in appendix \ref{sub:Proof-for-EquSINRopt}
and provides an algorithm to find the decomposition. }An immediate
consequence of Theorem \ref{thm:EquSINRopt} is a mapping of the Pareto
boundary points of the achievable rate region to the SINR region.
\begin{cor}
\textcolor{black}{\label{cor:Rate-SINRP}Let $M_{l}=\text{Rank}(\mathbf{H}_{l,l})$}%
\footnote{\textcolor{black}{This will not lose optimality because by Theorem
\ref{thm:WFST}, the rank of the optimal input covariance matrix for
link $l$ is no more than the rank of $\mathbf{H}_{l,l}$.}%
}\textcolor{black}{. An SINR point $\left[\gamma_{l,m}=e^{\mathcal{I}_{l}/M_{l}}-1\right]_{m=1,...,M_{l},l=1,...,L}$
is a Pareto boundary point in $\mathcal{T}_{\mathbf{\Phi}}\left(P_{T}\right)$,
if and only if the rate point $\left[\mathcal{I}_{l}\right]_{l=1,...,L}$
}is a Pareto rate point \textcolor{black}{in $\mathcal{R}_{\mathbf{\Phi}}\left(P_{T}\right)$.}
\end{cor}

\textcolor{black}{Therefore, the problem }with rate constraints\textcolor{black}{{}
$\left[\mathcal{I}_{l}^{0}\right]_{l=1,...,L}$ can be equivalently
solved through the problem with} SINR constraints\textcolor{black}{{}
$\left[\gamma_{l,m}=e^{\mathcal{I}_{l}/M_{l}}-1\right]_{m=1,...,M_{l},l=1,...,L}$. }

\textcolor{black}{The following theorem shows that uniform power allocation
across the streams within a link will also not lose optimality, which
is useful in designing algorithms for individual power constraints
and/or distributed optimization \cite{LiuAn_globecom09_IFCduality,AnLiu_Allerton09_Duality}.}
\begin{thm}
\textcolor{black}{\label{thm:equpoweropt}For any input covariance
matrix $\mathbf{\Sigma}$, there exists a decomposition $\mathbf{\Sigma}=\sum_{m=1}^{M}p_{m}\mathbf{t}_{m}\mathbf{t}_{m}^{\dagger}$
such that the transmit power is uniformly allocated over the $M$
streams, i.e., $p_{m}=\textrm{Tr}\left(\mathbf{\Sigma}\right)/M,\forall m$.
Therefore uniform power allocation over the streams of the same link
will not lose optimality. }
\end{thm}

The proof is given in Appendix \ref{sub:Proof-for-Theorem-Equpow}.

\subsection{\textcolor{black}{SINR based Algorithms}\textup{\textcolor{black}{\label{sub:Algorithms-to-Solve}}}}

\textcolor{black}{The results in Section \ref{sec:Rate-SINR} serve
as a bridge to solve the }\textbf{\textcolor{black}{FOP}}\textcolor{black}{{}
or }\textbf{\textcolor{black}{SPMP}}\textcolor{black}{{} under rate
constraints through the  SINR optimization problems. First we show
}\textbf{\textcolor{black}{FOP}}\textcolor{black}{{} is equivalent to
the following SINR optimization problem in the sense of feasibility.\begin{align}
\textrm{\textbf{EFOP}}:\underset{\left\{ \mathbf{T},\mathbf{R},\mathbf{p}\right\} }{\textrm{max}} & \min_{\begin{array}{c}
1\leq m\leq M_{l}\\
1\leq l\leq L\end{array}}\frac{\gamma_{l,m}}{\gamma_{l}^{0}},\:\textrm{s.t.}\left\Vert \mathbf{p}\right\Vert _{1}\leq P_{T}\label{eq:P1SINR}\end{align}
where $M_{l}=\textrm{Rank}\left(\mathbf{H}_{l,l}\right)$ is the number
of streams of link $l$; $\gamma_{l}^{0}=e^{\mathcal{I}_{l}^{0}/M_{l}}-1$
is the target SINR for the streams of link $l$.}
\begin{thm}
\textcolor{black}{\label{thm:P1equ}The optimum of }\textbf{\textcolor{black}{FOP}}\textcolor{black}{{}
(\ref{eq:P1}) is not less than 1 if and only if the optimum of }\textbf{\textcolor{black}{EFOP}}\textcolor{black}{{}
(\ref{eq:P1SINR}) is not less than 1.}\end{thm}
\begin{proof}
\textcolor{black}{If the optimum of }\textbf{\textcolor{black}{EFOP}}\textcolor{black}{{}
is not less than 1, there exists a point $\left[\gamma_{l,m}\geq\gamma_{l}^{0}\right]_{m=1,...,M_{l},l=1,...,L}$
in $\mathcal{T}_{\mathbf{\Phi}}\left(P_{T}\right)$. Then it follows
from Corollary \ref{cor:Rate-SINRP} that the rate point $\left[\mathcal{I}_{l}=M_{l}\textrm{log}\left(1+\gamma_{l,m}\right)\geq\mathcal{I}_{l}^{0}\right]_{l=1,...,L}$
lies in $\mathcal{R}_{\mathbf{\Phi}}\left(P_{T}\right)$, i.e., the
optimum of }\textbf{\textcolor{black}{FOP }}\textcolor{black}{is not
less than 1. The 'only if' part can be proved similarly.}
\end{proof}

\begin{remrk}
If \textcolor{black}{the target rates $\left(\mathcal{I}_{1}^{0},...,\mathcal{I}_{L}^{0}\right)$
is not a Pareto point but feasible, the solution of }\textbf{\textcolor{black}{EFOP}}\textcolor{black}{{}
may produce a Pareto rate point that is not a solution of the }\textbf{\textcolor{black}{FOP}}\textcolor{black}{.}\textbf{\textcolor{black}{{}
}}\textcolor{black}{However, both will exceed the target rates.}
\end{remrk}

\textcolor{black}{Similarly, }\textbf{\textcolor{black}{SPMP}}\textcolor{black}{{}
is equivalent to the following SINR optimization problem\begin{equation}
\textrm{\textbf{ESPMP}}:\underset{\left\{ \mathbf{T},\mathbf{R},\mathbf{p}\right\} }{\textrm{min}}\left\Vert \mathbf{p}\right\Vert _{1},\:\textrm{s.t.}\:\gamma_{l,m}\geq\gamma_{l}^{0},\begin{array}{c}
1\leq m\leq M_{l},\\
1\leq l\leq L.\end{array}\label{eq:P2SINR}\end{equation}
}
\begin{thm}
\textcolor{black}{\label{thm:P2equ}If $\left\{ \tilde{\mathbf{T}},\tilde{\mathbf{R}},\tilde{\mathbf{p}}\right\} $
is an optimum of} \textbf{ESPMP}\textcolor{black}{{} (\ref{eq:P2SINR}),
the input covariance matrices $\tilde{\mathbf{\Sigma}}_{1:L}$ produced
by $\tilde{\mathbf{T}}$ and $\tilde{\mathbf{p}}$ must be an optimum
of }\textbf{\textcolor{black}{SPMP}}\textcolor{black}{{} (\ref{eq:P2}).
On the other hand, if $\tilde{\mathbf{\Sigma}}_{1:L}$ is an optimum
of }\textbf{\textcolor{black}{SPMP}}\textcolor{black}{, there exists
a decomposition leading to $\left\{ \tilde{\mathbf{T}},\tilde{\mathbf{R}},\tilde{\mathbf{p}}\right\} $,
which is an optimum of} \textbf{ESPMP}.\end{thm}
\begin{proof}
\textcolor{black}{If $\tilde{\mathbf{\Sigma}}_{1:L}$ produced by
$\tilde{\mathbf{T}}$ and $\tilde{\mathbf{p}}$ is not an optimum
of }\textbf{\textcolor{black}{SPMP}}\textcolor{black}{, there exists
a solution $\mathbf{\Sigma}_{1:L}^{'}$ such that $\left[\mathcal{I}_{l}^{0}\right]_{l=1,...,L}$
can be achieved by a smaller total transmit power. Then it follows
from Theorem \ref{thm:EquSINRopt} that there exists a decomposition
of $\mathbf{\Sigma}_{1:L}^{'}$ such that the target SINRs $\left[\gamma_{l,m}^{0}=e^{\mathcal{I}_{l}^{0}/M_{l}}-1\right]_{m=1,...,M_{l},l=1,...,L}$
are achieved by the corresponding $\left\{ \mathbf{T}^{'},\mathbf{R}^{'},\mathbf{p}^{'}\right\} $
with $\left\Vert \mathbf{p}^{'}\right\Vert _{1}<\left\Vert \mathbf{\tilde{p}}\right\Vert _{1}$,
which contradicts with the optimality of $\left\{ \tilde{\mathbf{T}},\tilde{\mathbf{R}},\tilde{\mathbf{p}}\right\} $.
The reverse part can be proved similarly.}\end{proof}
\begin{remrk}
\textcolor{black}{It is possible that neither }\textbf{\textcolor{black}{SPMP}}\textcolor{black}{{}
nor }\textbf{\textcolor{black}{ESPMP}}\textcolor{black}{{} is feasible,
i.e., the target rates or SINRs can not be achieved even with infinite
power. In this paper, we only consider feasible problems. In practice,
we can avoid solving the infeasible }\textbf{\textcolor{black}{SPMP}}\textcolor{black}{{}
and }\textbf{\textcolor{black}{ESPMP}}\textcolor{black}{{} by solving
the }\textbf{\textcolor{black}{FOP}}\textcolor{black}{{} with the maximum
allowable sum power to test the feasibility first.}
\end{remrk}

\textcolor{black}{By converting the }\textbf{\textcolor{black}{FOP}}\textcolor{black}{{}
and }\textbf{\textcolor{black}{SPMP}}\textcolor{black}{{} to the simpler
}\textbf{\textcolor{black}{EFOP}}\textcolor{black}{{} and }\textbf{\textcolor{black}{ESPMP}}\textcolor{black}{,
these problems can be efficiently solved by the SINR duality based
algorithms for MIMO beamforming networks \cite{Martin_ITV_04_BFdual,Rao_TOC07_netduality},
which is summarized below. For fixed $\mathbf{T},\mathbf{p}$, the
optimal receive vector for each stream is decoupled and is given by
the MMSE-SIC receiver in (\ref{eq:MMSErev1G}). The SINR duality in
Lemma \ref{lem:lem1G} implies that the coupled problem of optimizing
the transmit vectors can be found by optimizing the receive vectors
in the reverse links, i.e., for fixed $\mathbf{R},\mathbf{q}$, the
optimal receive vector for each stream in the reverse links is given
by the MMSE-SIC receiver\begin{align}
\mathbf{t}_{l,m} & =\beta_{l,m}\left(\sum_{i=1}^{m-1}q_{l,i}\mathbf{H}_{l,l}^{\dagger}\mathbf{r}_{l,i}\mathbf{r}_{l,i}^{\dagger}\mathbf{H}_{l,l}+\hat{\mathbf{\Omega}}_{l}\right)^{-1}\mathbf{H}_{l,l}^{\dagger}\mathbf{r}_{l,m}\label{eq:TransmitOpt}\end{align}
where $\hat{\mathbf{\Omega}}_{l}$ is obtained from $\hat{\mathbf{\Sigma}}_{k}=\sum_{i=1}^{M_{k}}q_{k,i}\mathbf{r}_{k,i}\mathbf{r}_{k,i}^{\dagger},\ k\ne l$
using (\ref{eq:WhiteMRV}), and $\beta_{l,m}$ is chosen such that
$\left\Vert \mathbf{t}_{l,m}\right\Vert =1$.}

\textcolor{black}{The optimization for $\mathbf{p}$ is different
for the two problems. The algorithm designed in \cite{Martin_ITV_04_BFdual}
for }\textbf{\textcolor{black}{EFOP}}\textcolor{black}{{} is described
below. For fixed $\mathbf{T},\mathbf{R}$, define the following extended
coupling matrices \cite{Martin_ITV_04_BFdual} \begin{eqnarray}
\mathbf{\Upsilon} & = & \left[\begin{array}{cc}
\mathbf{D}\mathbf{\Psi} & \mathbf{D}\mathbf{1}\\
\frac{1}{P_{T}}\mathbf{1}^{T}\mathbf{D}\mathbf{\Psi} & \frac{1}{P_{T}}\mathbf{1}^{T}\mathbf{D}\mathbf{1}\end{array}\right],\label{eq:ExtendcpMG1}\end{eqnarray}
\begin{eqnarray}
\mathbf{\mathbf{\Lambda}} & = & \left[\begin{array}{cc}
\mathbf{D}\mathbf{\Psi}{}^{T} & \mathbf{D}\mathbf{1}\\
\frac{1}{P_{T}}\mathbf{1}^{T}\mathbf{D}\mathbf{\Psi}{}^{T} & \frac{1}{P_{T}}\mathbf{1}^{T}\mathbf{D}\mathbf{1}\end{array}\right],\label{eq:ExtendcpMG2}\end{eqnarray}
where $\mathbf{\Psi}$ and $\mathbf{D}$ are defined in (\ref{eq:faiG})
and (\ref{eq:DG}) respectively; and the SINR values in $\mathbf{D}$
are fixed as $\gamma_{l,m}^{0}=\gamma_{l}^{0},m=1,...,M_{l},l=1,...,L$,
where $\gamma_{l}^{0}$'s are the target SINRs in }\textbf{\textcolor{black}{EFOP}}\textcolor{black}{.
With the optimal power $\mathbf{\tilde{p}}$, all the scaled SINRs
in (\ref{eq:P1SINR}) should be equal to the same value denoted as
$C_{\textrm{max}}$ \cite{Martin_ITV_04_BFdual}. Therefore $\mathbf{\tilde{p}}$
satisfies the equations $\gamma_{l,m}=C_{\textrm{max}}\gamma_{l}^{0},m=1,...,M_{l},l=1,...,L$
and $\left\Vert \tilde{\mathbf{p}}\right\Vert _{1}=P_{T}$, which
together form the following eigensystem \cite{Martin_ITV_04_BFdual}}%
{}\textcolor{black}{\begin{eqnarray}
\mathbf{\mathbf{\Upsilon}}\mathbf{p}_{\textrm{ext}} & = & \lambda_{\textrm{max}}\mathbf{p}_{\textrm{ext}},\label{eq:eigenP2}\end{eqnarray}
where $\mathbf{p}_{\textrm{ext}}=\left[\mathbf{\tilde{p}}^{T},1\right]^{T}$
is the dominant eigenvector of $\mathbf{\mathbf{\Upsilon}}$ with
its last component scaled to one; $\lambda_{\textrm{max}}=1/C_{\textrm{max}}$
is the corresponding maximum eigenvalue. It was proved in \cite{Wyang_ICASSP_98_singleSolNNM}
that for a non-negative matrix $\mathbf{\Upsilon}$ with the special
structure (\ref{eq:ExtendcpMG1}), the maximal eigenvalue and its
associated eigenvector are strictly positive, and no other eigenvalue
fulfills the positivity requirement. Therefore, the last component
of $\mathbf{p}_{\textrm{ext}}$ can always be scaled to one and the
resulting $\mathbf{\tilde{p}}$ is a valid power vector. Similarly,
in the reverse links, the optimal power $\mathbf{\tilde{q}}$ is obtained
by solving the eigensystem below \cite{Martin_ITV_04_BFdual}\begin{equation}
\mathbf{\Lambda}\mathbf{q}_{\textrm{ext}}=\hat{\lambda}_{\textrm{max}}\mathbf{q}_{\textrm{ext}},\label{eq:eigenP1}\end{equation}
where $\mathbf{q}_{\textrm{ext}}=\left[\mathbf{\tilde{q}}^{T},1\right]^{T}$
is the dominant eigenvector of $\mathbf{\Lambda}$; and $\hat{\lambda}_{\textrm{max}}$
is the corresponding maximum eigenvalue.}

\textcolor{black}{For }\textbf{\textcolor{black}{ESPMP}}\textcolor{black}{,
we use the algorithm in \cite{Rao_TOC07_netduality}. Let $\gamma_{l,m}^{(n)}$
and $\hat{\gamma}_{l,m}^{(n)}$ respectively be the SINR for the $m^{\text{th}}$
stream of the forward and reverse link $l$ after the $n^{\text{th}}$
update. }%
{}\textcolor{black}{To satisfy the SINR constraints, the standard power
control is used to update $\mathbf{p}$ and $\mathbf{q}$ iteratively,
where in each iteration, the power of the stream with over-satisfied
(unsatisfied) SINR is reduced (increased): \begin{eqnarray}
p_{l,m}^{(n+1)} & = & \frac{\gamma_{l}^{0}}{\gamma_{l,m}^{(n)}}p_{l,m}^{(n)},\label{eq:Upppower}\\
q_{l,m}^{(n+1)} & = & \frac{\gamma_{l}^{0}}{\hat{\gamma}_{l,m}^{(n)}}q_{l,m}^{(n)}.\label{eq:Upqpower}\end{eqnarray}
For convenience, we rewrite (\ref{eq:Upppower}) and (\ref{eq:Upqpower})
into vector functions\begin{eqnarray*}
\mathbf{p}^{(n+1)} & = & \mathbf{I}_{P}\left(\mathbf{p}^{(n)}\right),\\
\mathbf{q}^{(n+1)} & = & \mathbf{I}_{D}\left(\mathbf{q}^{(n)}\right).\end{eqnarray*}
}

\begin{table}
\caption{\label{tab:table1}Algorithm A (Solving\textbf{ EFOP})}

\centering{}\begin{tabular}{l}
\hline
\textbf{\small Choose}{\small{} $\mathbf{p}^{(0)}>0$ and $\mathbf{T}^{(0)}$
such that all streams have positive SINR.}\tabularnewline
\textbf{\small Set}{\small{} $n\leftarrow0$.}\tabularnewline
\textbf{\small While}{\small{} not converge }\textbf{\small do}{\small{} }\tabularnewline
{\small $\;$1. Update in the forward links }\tabularnewline
{\small $\;$$\;$a. Compute $\mathbf{R}^{(n+1)}$ from $\mathbf{p}^{(n)}$
and $\mathbf{T}^{(n)}$ using (\ref{eq:MMSErev1G}); }\tabularnewline
{\small $\;$$\;$b. Compute $\mathbf{\Lambda}$ from $\mathbf{T}^{(n)}$
and $\mathbf{R}^{(n+1)}$; }\tabularnewline
{\small $\;$$\;$c. Solve eigensystem (\ref{eq:eigenP1}) for $\mathbf{q}^{(n+1)}$; }\tabularnewline
{\small $\;$2. Update in the reverse links }\tabularnewline
{\small $\;$$\;$a. Compute $\mathbf{T}^{(n+1)}$ from $\mathbf{q}^{(n+1)}$
and $\mathbf{R}^{(n+1)}$ using (\ref{eq:TransmitOpt}); }\tabularnewline
{\small $\;$$\;$b. Compute $\mathbf{\mathbf{\Upsilon}}$ from $\mathbf{T}^{(n+1)}$
and $\mathbf{R}^{(n+1)}$;}\tabularnewline
{\small $\;$$\;$c. Solve eigensystem (\ref{eq:eigenP2}) for $\mathbf{p}^{(n+1)}$; }\tabularnewline
{\small $\;$$n\leftarrow n+1$; }\tabularnewline
\textbf{\small End}\tabularnewline
\end{tabular}
\end{table}

\begin{table}
\caption{\label{tab:table2}Algorithm B (Solving\textbf{ ESPMP})}

\centering{}\begin{tabular}{l}
\hline
\textbf{\small Choose}{\small{} $\mathbf{p}^{(0)}>0$ and $\mathbf{T}^{(0)}$
such that all streams have positive SINR. }\tabularnewline
\textbf{\small Set}{\small{} $n\leftarrow0$. }\tabularnewline
\textbf{\small While}{\small{} not converge }\textbf{\small do}{\small{} }\tabularnewline
{\small $\;$1. Update in the forward links }\tabularnewline
{\small $\;$$\;$a. Compute $\mathbf{R}^{(n+1)}$ from $\mathbf{p}^{(n)}$
and $\mathbf{T}^{(n)}$ using (\ref{eq:MMSErev1G}); }\tabularnewline
{\small $\;$$\;$b. Update power in the forward links $\mathbf{p}^{(n+1)}=\mathbf{I}_{P}\left(\mathbf{p}^{(n)}\right)$; }\tabularnewline
{\small $\;$$\;$c. Let $\left\{ \gamma_{l,m}^{0}=\gamma_{l,m}\left(\mathbf{T}^{(n)},\mathbf{R}^{(n+1)},\mathbf{p}^{(n+1)}\right),\textrm{for\:\ all}\: m,l\right\} $; }\tabularnewline
{\small $\;$$\;$d. Compute $\mathbf{D}$ and $\mathbf{\Psi}$ from
$\mathbf{T}^{(n)}$, $\mathbf{R}^{(n+1)}$ and $\left\{ \gamma_{l,m}^{0}\right\} $; }\tabularnewline
{\small $\;$$\;$e. Compute $\mathbf{q}^{(n)}=\left(\mathbf{D}^{-1}-\mathbf{\Psi}^{T}\right)^{-1}\mathbf{1}$; }\tabularnewline
{\small $\;$2. Update in the reverse links }\tabularnewline
{\small $\;$$\;$a. Compute $\mathbf{T}^{(n+1)}$ from $\mathbf{q}^{(n)}$
and $\mathbf{R}^{(n+1)}$ using (\ref{eq:TransmitOpt}); }\tabularnewline
{\small $\;$$\;$b. Update power in the forward links $\mathbf{q}^{(n+1)}=\mathbf{I}_{D}\left(\mathbf{q}^{(n)}\right)$; }\tabularnewline
{\small $\;$$\;$c. Let $\left\{ \gamma_{l,m}^{0}=\gamma_{l,m}\left(\mathbf{T}^{(n+1)},\mathbf{R}^{(n+1)},\mathbf{q}^{(n+1)}\right),\textrm{for\:\ all}\: m,l\right\} $; }\tabularnewline
{\small $\;$$\;$d. Compute $\mathbf{D}$ and $\mathbf{\Psi}$ from
$\mathbf{T}^{(n+1)}$, $\mathbf{R}^{(n+1)}$ and $\left\{ \gamma_{l,m}^{0}\right\} $; }\tabularnewline
{\small $\;$$\;$e. Compute $\mathbf{p}^{(n+1)}=\left(\mathbf{D}^{-1}-\mathbf{\Psi}\right)^{-1}\mathbf{1}$; }\tabularnewline
{\small $\;$$n\leftarrow n+1$; }\tabularnewline
\textbf{\small End}\tabularnewline
\end{tabular}
\end{table}

\textcolor{black}{The algorithms to solve }\textbf{\textcolor{black}{EFOP}}\textcolor{black}{{}
and }\textbf{\textcolor{black}{ESPMP}}\textcolor{black}{{} are summarized
in table \ref{tab:table1} and table \ref{tab:table2} respectively.
After obtaining $\mathbf{T},\mathbf{R}$ and $\mathbf{p}$, the corresponding
input covariance matrices for }\textbf{\textcolor{black}{FOP}}\textcolor{black}{{}
and }\textbf{\textcolor{black}{SPMP}}\textcolor{black}{{} can be easily
obtained}\textbf{\textcolor{black}{.}}\textcolor{black}{{} }

\textcolor{black}{The convergence of these algorithms are proved in
\cite{Martin_ITV_04_BFdual,Rao_TOC07_netduality}. It can be verified
that the objective function in }\textbf{\textcolor{black}{EFOP}}\textcolor{black}{{}
(\ref{eq:P1SINR}) is monotonically increased by Algorithm A \cite{Martin_ITV_04_BFdual}.
In Algorithm B,  once the solution becomes feasible, i.e., all SINR
values meet or exceed the minimum requirements, it generates a sequence
of feasible solutions with monotonically decreasing sum power \cite{Rao_TOC07_netduality}.
The optimality of these algorithms will be discussed in the next section.}

\subsection{\textcolor{black}{\label{sub:Optimality-Analysis}Optimality Analysis
for SINR based Algorithms}}

\textcolor{black}{}%
{}\textcolor{black}{Algorithm A or B can find good solutions but may
not find the optimum for general B-MAC networks, according to the
numeric examples in Section \ref{sec:Simulation-Results}. But we
can still obtain insight of the problem and derive improved algorithms
by finding the necessary conditions satisfied by the optimum.}

\textcolor{black}{To avoid deriving the necessary conditions with
the non-differentiable objective function of }\textbf{\textcolor{black}{FOP}}\textcolor{black}{{}
(\ref{eq:P1}), we rewrite it into the following equivalent problem\begin{eqnarray}
\textrm{\textbf{FOPa}:} & \underset{\mathbf{\Sigma}_{1:L}}{\textrm{max}} & \frac{\mathcal{I}_{1}\left(\mathbf{\Sigma}_{1:L},\mathbf{\Phi}\right)}{\mathcal{I}_{1}^{0}}\label{eq:P1a}\\
 & \textrm{s.t.} & \frac{\mathcal{I}_{l}\left(\mathbf{\Sigma}_{1:L},\mathbf{\Phi}\right)}{\mathcal{I}_{l}^{0}}\geq\frac{\mathcal{I}_{1}\left(\mathbf{\Sigma}_{1:L},\mathbf{\Phi}\right)}{\mathcal{I}_{1}^{0}},\forall l\neq1\nonumber \\
 &  & \mathbf{\Sigma}_{l}\succeq0,l=1,\cdots,L\:\textrm{and}\:\sum_{l=1}^{L}\textrm{Tr}\left(\mathbf{\Sigma}_{l}\right)\leq P_{T}.\nonumber \end{eqnarray}
Then the following theorem holds.}
\begin{thm}
\textcolor{black}{\label{thm:optimality}Necessity: If $\tilde{\mathbf{\Sigma}}_{1:L}=\left(\mathbf{\tilde{\mathbf{\Sigma}}}_{1},...,\mathbf{\tilde{\mathbf{\Sigma}}}_{L}\right)$
is an optimum of }\textbf{\textcolor{black}{FOPa}}\textcolor{black}{{}
(\ref{eq:P1a}) or }\textbf{\textcolor{black}{SPMP}}\textcolor{black}{{}
(\ref{eq:P2}), it must satisfy the optimality conditions below: }
\begin{enumerate}
\item \textcolor{black}{It possesses the polite water-filling structure
as in Definition \ref{def:DefGWF}. }
\item \textcolor{black}{The achieved rates must satisfy $\mathcal{I}_{l}\left(\tilde{\mathbf{\Sigma}}_{1:L},\mathbf{\Phi}\right)=\alpha\mathcal{I}_{l}^{0},l=1,...,L$,
where for }\textbf{\textcolor{black}{FOPa}}\textcolor{black}{, $\alpha>0$
is some constant; and for }\textbf{\textcolor{black}{SPMP}}\textcolor{black}{,
$\alpha=1$. }
\item For \textbf{FOPa}, \textcolor{black}{it satisfies $\sum_{l=1}^{L}\textrm{Tr}\left(\tilde{\mathbf{\Sigma}}_{l}\right)=P_{T}$.}
\end{enumerate}
\textcolor{black}{On the other hand, if certain $\mathbf{\tilde{\Sigma}}_{1:L}$
satisfies the above optimality conditions for }\textbf{\textcolor{black}{FOPa}}\textcolor{black}{{}
or }\textbf{\textcolor{black}{SPMP}}\textcolor{black}{, it must satisfy
the }Karush\textendash{}Kuhn\textendash{}Tucker (KKT)\textcolor{black}{{}
conditions of }\textbf{\textcolor{black}{FOPa }}\textcolor{black}{or
}\textbf{\textcolor{black}{SPMP}}\textcolor{black}{, and thus achieves
a stationary point. }

\textcolor{black}{Sufficiency: If certain $\mathbf{\tilde{\Sigma}}_{1:L}$
satisfies the above optimality conditions for }\textbf{\textcolor{black}{FOPa}}\textcolor{black}{{}
or }\textbf{\textcolor{black}{SPMP}}\textcolor{black}{{} and if the
weighted sum rate $\sum_{l}^{L}\tilde{\nu}_{l}\mathcal{I}_{l}\left(\mathbf{\Sigma}_{1:L},\mathbf{\Phi}\right)$
is a concave function of $\mathbf{\Sigma}_{1:L}$, where $\tilde{\nu}_{l}$'s
are the polite water-filling levels of $\tilde{\mathbf{\Sigma}}_{1:L}$
in (\ref{eq:WFFar}), then $\tilde{\mathbf{\Sigma}}_{1:L}$ is the
optimum of }\textbf{\textcolor{black}{FOPa}}\textcolor{black}{{} or
}\textbf{\textcolor{black}{SPMP}}\textcolor{black}{.}
\end{thm}

\textcolor{black}{It can be proved by contradiction that the optimums
of }\textbf{\textcolor{black}{FOPa}}\textcolor{black}{{} and }\textbf{\textcolor{black}{SPMP}}\textcolor{black}{{}
are Pareto optimal. By Theorem \ref{thm:WFST}, they possess the polite
water-filling structure. The second optimality condition can be proved
by a proof similar to that of Lemma 1 in \cite{Rao_TOC07_netduality}
for }\textbf{\textcolor{black}{ESPMP}}\textcolor{black}{. The third
optimality condition can also be proved by the contradiction that
if the total transmit power is less that $P_{T}$, the extra power
can be used to improve the rates of all links simultaneously. The
connection between the necessary optimality conditions and the KKT
conditions, and the sufficiency part are proved in appendix \ref{sub:Proof-of-optcon}. }

\textcolor{black}{}%
{}\textcolor{black}{We check whether the solutions of Algorithms A and
B satisfy the optimality conditions.}%
{}\textcolor{black}{{} We use the notation $\bar{}$ for the variables
corresponding to the solution of Algorithm A or B. The following is
obvious.}
\begin{lemma}
\textcolor{black}{\label{lem:after-convergence}After the convergence
of the Algorithm A or B, the following conditions are satisfied. }
\begin{enumerate}
\item \textcolor{black}{In the forward (reverse) links, the MMSE-SIC receive
vectors corresponding to $\bar{\mathbf{T}}$ and $\bar{\mathbf{p}}$
($\bar{\mathbf{R}}$ and $\mathbf{\bar{q}}$) are given by $\mathbf{\bar{R}}$
($\mathbf{\bar{T}}$). The set of SINRs achieved by $\left\{ \mathbf{\bar{T}},\mathbf{\bar{R}},\mathbf{\bar{p}}\right\} $
in the forward links equals to that achieved by $\left\{ \mathbf{\bar{R}},\mathbf{\bar{T}},\mathbf{\bar{q}}\right\} $
in the reverse links.}
\item \textcolor{black}{For Algorithm B, the achieved rates satisfy $\mathcal{I}_{l}\left(\bar{\mathbf{\Sigma}}_{1:L},\mathbf{\Phi}\right)=\mathcal{I}_{l}^{0},l=1,...,L$.}
\item For \textcolor{black}{Algorithm }A, \textcolor{black}{$\bar{\mathbf{\Sigma}}_{1:L}$
satisfies $\sum_{l=1}^{L}\textrm{Tr}\left(\bar{\mathbf{\Sigma}}_{l}\right)=P_{T}$.}
\end{enumerate}
\end{lemma}

\textcolor{black}{Note that the rates achieved by Algorithm A may
not satisfy the condition $\mathcal{I}_{l}\left(\bar{\mathbf{\Sigma}}_{1:L},\mathbf{\Phi}\right)=\alpha\mathcal{I}_{l}^{0},l=1,...,L$.
In order to discuss the optimality, we modify the target rates in
}\textbf{\textcolor{black}{FOP/FOPa}}\textcolor{black}{{} to $\mathcal{I}_{l}^{0}=\mathcal{I}_{l}\left(\bar{\mathbf{\Sigma}}_{1:L},\mathbf{\Phi}\right)$.
Then, we can claim that the solution of Algorithm A also satisfies
the second optimality condition. }

\textcolor{black}{Furthermore, if one can prove that the solution
$\bar{\mathbf{\Sigma}}_{1:L}$ possesses the polite water-filling
structure, then $\bar{\mathbf{\Sigma}}_{1:L}$ satisfies all optimality
conditions in Theorem \ref{thm:optimality}. One might conjecture
that the first condition on MMSE structure in Lemma \ref{lem:after-convergence}
implies the polite water-filling structure. Unfortunately, this is
not always true according to the following counter example. Consider
a single user channel $\mathbf{H}$ with $\textrm{Rank}\left(\mathbf{H}\right)>1$
and unequal non-zero singular values. If the transmit vectors are
initialized as the non-zero right singular vectors of $\mathbf{H}$,
the algorithm will converge to a solution where the transmit and receive
vectors respectively are the non-zero right and left singular vectors
of $\mathbf{H}$, and the transmit powers will make the SINRs of all
streams the same. Then the solution does not satisfy the single-user
water-filling structure. }%
{}\textcolor{black}{However, for a smaller class of channels, we have
the following. }
\begin{thm}
\textcolor{black}{\label{thm:PesoWF}If $\textrm{Rank}\left(\mathbf{H}_{l,l}\right)=1,\forall l$,
the solution of Algorithm A (B) $\bar{\mathbf{\Sigma}}_{1:L}$ satisfies
all the optimality conditions in Theorem \ref{thm:optimality}, and
thus achieves a stationary point.}\end{thm}
\begin{proof}
\textcolor{black}{The polite water-filling structure of $\bar{\mathbf{\Sigma}}_{1:L}$
can be proved by considering the transmission over the equivalent
channel $\bar{\mathbf{H}}_{l}\triangleq\mathbf{\bar{\Omega}}_{l}^{-1/2}\mathbf{H}_{l,l}\bar{\hat{\mathbf{\Omega}}}_{l}^{-1/2}$.
}Decompose the forward and reverse equivalent input covariance matrices
$\mathbf{\bar{Q}}_{l}\triangleq\bar{\hat{\mathbf{\Omega}}}_{l}^{1/2}\bar{\mathbf{\Sigma}}_{l}\bar{\hat{\mathbf{\Omega}}}_{l}^{1/2}$
and $\mathbf{\bar{\hat{Q}}}_{l}=\mathbf{\bar{\Omega}}_{l}^{1/2}\mathbf{\bar{\hat{\Sigma}}}_{l}\mathbf{\bar{\Omega}}_{l}^{1/2}$
to beams as $\mathbf{\bar{Q}}_{l}=\sum_{m=1}^{M_{l}}\bar{d}_{l,m}\bar{\mathbf{u}}_{l,m}\mathbf{\bar{u}}_{l,m}^{\dagger}$,
where $\bar{d}_{l,m}=\bar{p}_{l,m}\left\Vert \bar{\hat{\mathbf{\Omega}}}_{l}^{1/2}\mathbf{\bar{t}}_{l,m}\right\Vert ^{2}$
is the equivalent transmit power and $\mathbf{\bar{u}}_{l,m}=\bar{\hat{\mathbf{\Omega}}}_{l}^{1/2}\sqrt{\bar{p}_{l,m}}\mathbf{\bar{t}}_{l,m}/\sqrt{\bar{d}_{l,m}}$
is the equivalent transmit vector; and $\mathbf{\bar{\hat{Q}}}_{l}=\sum_{m=1}^{M_{l}}\bar{\hat{d}}_{l,m}\mathbf{\bar{v}}_{l,m}\mathbf{\bar{v}}_{l,m}^{\dagger}$,
where $\bar{\hat{d}}_{l,m}=\bar{q}_{l,m}\left\Vert \mathbf{\bar{\Omega}}_{l}^{1/2}\bar{\mathbf{r}}_{l,m}\right\Vert ^{2}$
and $\mathbf{\bar{v}}_{l,m}=\mathbf{\bar{\Omega}}_{l}^{1/2}\sqrt{\bar{q}_{l,m}}\mathbf{\bar{r}}_{l,m}/\sqrt{\bar{\hat{d}}_{l,m}}$.\textcolor{black}{{}
The algorithm sets the number of data streams as $M_{l}=\textrm{Rank}\left(\mathbf{H}_{l,l}\right)=1,\forall l$,
which does not lose optimality by Theorem \ref{thm:WFST}. Since the
interference-plus-noise is whitened in the equivalent channel $\bar{\mathbf{H}}_{l}$,
the MMSE receiver $\mathbf{\bar{v}}_{l,1}$ reduces to the matched
filter, i.e., $\mathbf{\bar{v}}_{l,1}=\alpha_{l,1}\bar{\mathbf{H}}_{l}\bar{\mathbf{u}}_{l,1}$.
Similarly, the MMSE receiver $\bar{\mathbf{u}}_{l,1}$ in the reverse
equivalent channel $\bar{\mathbf{H}}_{l}^{\dagger}$ is given by the
matched filter: $\bar{\mathbf{u}}_{l,1}=\beta_{l,1}\bar{\mathbf{H}}_{l}^{\dagger}\bar{\mathbf{v}}_{l,1}=\alpha_{l,1}\beta_{l,1}\bar{\mathbf{H}}_{l}^{\dagger}\bar{\mathbf{H}}_{l}\bar{\mathbf{u}}_{l,1}$,
i.e., $\bar{\mathbf{u}}_{l,1}$ is an eigenvector of $\bar{\mathbf{H}}_{l}^{\dagger}\bar{\mathbf{H}}_{l}$.
Since the initial point is chosen such that the SINRs of all streams
are strictly positive, they must also be strictly positive after the
convergence. Hence, $\bar{\mathbf{u}}_{l,1}$ must be the eigenvector
corresponding to the only non-zero eigenvalue $\delta_{l,1}$ of $\bar{\mathbf{H}}_{l}^{\dagger}\bar{\mathbf{H}}_{l}$.
Then $\mathbf{\bar{Q}}_{l}=\bar{d}_{l,1}\mathbf{\bar{u}}_{l,1}\mathbf{\bar{u}}_{l,1}^{\dagger}=\mathbf{u}_{l,1}\left(\bar{\nu}_{l}-1/\delta_{l,1}\right)^{+}\mathbf{u}_{l,1}^{\dagger}$,
where $\bar{\nu}_{l}\triangleq\bar{d}_{l,1}+1/\delta_{l,1}$ is the
polite water-filling level. Therefore, the solution $\bar{\mathbf{\Sigma}}_{1:L}$
satisfies the polite water-filling structure and all other optimality
conditions by Lemma \ref{lem:after-convergence}.}
\end{proof}

\textcolor{black}{}%
{}\textcolor{black}{For general cases, Algorithms A and B may converge
to a solution $\bar{\mathbf{\Sigma}}_{1:L}$ where some $\bar{\mathbf{\Sigma}}_{l}$'s
do not possess the polite water-filling structure. Then the rates
of these links may be improved without hurting other links by enforcing
the polite water-filling structure on these $\bar{\mathbf{\Sigma}}_{l}$'s
. In the next sub-section, we show how to do it by improving the algorithms
for a sub-class of B-MAC networks}\textit{\textcolor{black}{\emph{
named iTree Network}}}\textcolor{black}{s.}

\subsection{Improved Algorithm for iTree Network\textup{s\label{sub:itree}}}

iTree networks defined in \cite{Liu_IT10s_Duality_BMAC} appears to
be a natural extension of MAC and BC. We review its definition below.

\begin{figure}
\begin{centering}
\textsf{\includegraphics[clip,scale=0.3]{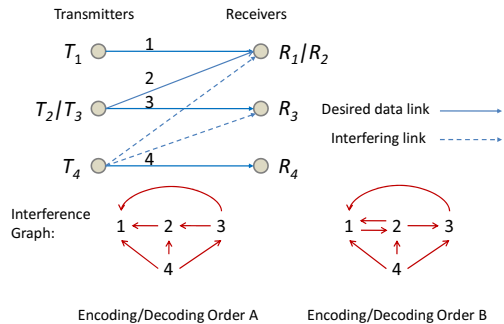}}
\par\end{centering}

\caption{\label{fig:iTree}A network with topology loop and encoding/decoding
order A is an iTree network, whose interference graph does not have
any directional loop. With encoding/decoding order B, it is not an
iTree network because the interference graph has directional loops.}

\end{figure}

\begin{definitn}
A B-MAC network with a fixed coupling matrix is called an \emph{Interference
Tree (iTree) Network} if after interference cancellation, the links
can be indexed such that any link is not interfered by the links with
smaller indices.
\end{definitn}

\begin{definitn}
In an \emph{Interference Graph}, each node represents a link. A directional
edge from node $i$ to node $j$ means that link $i$ causes interference
to link $j$.\end{definitn}
\begin{remrk}
The iTree network is related to but different from the network with
tree topology, which implies iTree network only if the interference
cancellation order is chosen properly. For example, a MAC which has
tree topology is not an iTree network if the successive decoding is
not employed at the receiver. On the other hand, even if there are
loops in a network, it may be an iTree network if the interference
cancellation order is right. We give such an example in Fig. \ref{fig:iTree}
where there are four desired data links 1, 2, 3, and 4, and dirty
paper coding and successive decoding and cancellation are employed.
With encoding/decoding order A, where the signal $\mathbf{x}_{2}$
is decoded after $\mathbf{x}_{1}$ and the signal $\mathbf{x}_{3}$
is encoded after $\mathbf{x}_{2}$, each link $l\in\left\{ 2,3,4\right\} $
is not interfered by the first $l-1$ links. Therefore, the network
in Fig. \ref{fig:iTree} is an iTree network even though it has a
loop of nonzero channel gains. However, for encoding/decoding order
B, SIC is not employed at $R_{1}/R_{2}$, and $\mathbf{x}_{2}$ is
encoded after $\mathbf{x}_{3}$ at $T_{2}/T_{3}$. The network in
Fig. \ref{fig:iTree} is no longer an iTree network because the interference
graph has directional loops, making the iTree indexing impossible.
\end{remrk}

Since the coupling matrix of the reverse links is the transpose of
that of the forward links, the interference relation is reversed as
stated in the following lemmas. Without loss of generality, we consider
iTree networks where the $l^{th}$ link is not interfered by the first
$l-1$ links in this paper.
\begin{lemma}
\label{lem:RevInf} \cite{Liu_IT10s_Duality_BMAC} If in an iTree
network, the $l^{th}$ link is not interfered by the links with lower
indices, in the reverse links, the $l^{th}$ link is not interfered
by the links with higher indices.
\end{lemma}

We develop an algorithm to improve the performance of iTree networks.
That is after the convergence of the Algorithm A or B, if any output
of the algorithm $\bar{\mathbf{\Sigma}}_{i}$ does not satisfy the
\textcolor{black}{polite water-filling structure, the objective (cost)
in }\textbf{\textcolor{black}{FOP}}\textcolor{black}{{} (}\textbf{\textcolor{black}{SPMP}}\textcolor{black}{)
can be strictly increased (decreased)} by enforcing this structure\textcolor{black}{{}
at link $i$. }

We first define some notations and give a useful lemma.\textcolor{black}{{}
The output of the SINR based algorithm $\mathbf{\bar{\Sigma}}_{1:L}$
achieves a rate point $\left[\mathcal{I}_{l}\right]_{l=1,...,L}$
with sum power $P_{T}\triangleq\sum_{l=1}^{L}\textrm{Tr}\left(\mathbf{\bar{\Sigma}}_{l}\right)$.
The algorithm also produces the corresponding covariance transformation
$\bar{\mathbf{\hat{\mathbf{\Sigma}}}}_{1:L}$ computed from $\mathbf{\bar{R}}$
and $\mathbf{\bar{q}}$ achieving a set of rates $\hat{\mathcal{I}}_{l}=\mathcal{I}_{l},l=1,...,L$.
Fixing the input covariance matrices }$\bar{\mathbf{\Sigma}}_{j},j=i+1,...,L$
for the last $L-i$ links, the first $i$ links form a sub-network\textcolor{black}{\begin{equation}
\left(\left[\mathbf{H}_{l,k}\right]_{k,l=1,...,i},\sum_{l=1}^{i}\textrm{Tr}\left(\mathbf{\Sigma}_{l}\right)=P_{T}^{i},\left[\mathbf{W}_{l}\right]_{l=1,...i}\right),\label{eq:Sub1toi}\end{equation}
where $\mathbf{W}_{l}=\mathbf{I}+\sum_{j=i+1}^{L}\mathbf{\Phi}_{l,j}\mathbf{H}_{l,j}\mathbf{\bar{\Sigma}}_{j}\mathbf{H}_{l,j}^{\dagger},\forall l$}
\textcolor{black}{is the covariance matrix of the equivalent colored
noise; $P_{T}^{i}=\sum_{l=1}^{i}\textrm{Tr}\left(\mathbf{\bar{\Sigma}}_{l}\right)$.
By Theorem \ref{thm:linear-color-dual}, the dual sub-network is \begin{equation}
\left(\left[\mathbf{H}_{k,l}^{\dagger}\right]_{k,l=1,...,i},\sum_{l=1}^{i}\textrm{Tr}\left(\hat{\mathbf{\Sigma}}_{l}\mathbf{W}_{l}\right)=P_{T}^{i},\left[\mathbf{I}_{l}\right]_{l=1,...,i}\right).\label{eq:Sub1toi-dual}\end{equation}
It is clear that after convergence, $\bar{\mathbf{\hat{\mathbf{\Sigma}}}}_{1:L}=\left(\mathbf{\bar{\hat{\Sigma}}}_{1},...,\mathbf{\bar{\hat{\Sigma}}}_{L}\right)$
is the covariance transformation of $\mathbf{\bar{\Sigma}}_{1:L}=\left(\mathbf{\bar{\Sigma}}_{1},...,\mathbf{\bar{\Sigma}}_{L}\right)$.
By Lemma 9 in \cite{Liu_IT10s_Duality_BMAC}, }$\bar{\hat{\mathbf{\Sigma}}}_{1:i}=\left(\mathbf{\bar{\hat{\Sigma}}}_{1},...,\mathbf{\bar{\hat{\Sigma}}}_{i}\right)$
is also the covariance transformation of $\mathbf{\bar{\Sigma}}_{1:i}=\left(\mathbf{\bar{\Sigma}}_{1},...,\mathbf{\bar{\Sigma}}_{i}\right)$,
applied to the sub-network (\ref{eq:Sub1toi}).

\textcolor{black}{The algorithm} to improve the performance contains
three steps.

\textbf{\textcolor{black}{Step}}\textcolor{black}{{} 1: Improve the
rate of reverse link $i$ by enforcing the polite water-filling structure
on $\hat{\mathbf{\Sigma}}_{i}$. By Lemma \ref{lem:RevInf}, the reverse
link $i$ causes no interference to the first $i-1$ reverse links.
If we fix $\bar{\hat{\mathbf{\Sigma}}}_{l},l=1,...,i-1$, the rate
of reverse link $i$ can be improved without hurting other reverse
links in the sub network by solving the following single-user optimization
problem:}

\textcolor{black}{\begin{align}
 & \underset{\hat{\mathbf{\Sigma}}_{i}\geq0}{\textrm{max}}\:\textrm{log}\left|\mathbf{I}+\mathbf{H}_{i,i}^{\dagger}\hat{\mathbf{\Sigma}}_{i}\mathbf{H}_{i,i}\bar{\mathbf{\hat{\mathbf{\Omega}}}}_{i}^{-1}\right|\label{eq:MRSC}\\
\textrm{s.t.} & \:\textrm{Tr}\left(\hat{\mathbf{\Sigma}}_{i}\mathbf{W}_{i}\right)\leq P_{T}^{i}-\sum_{l=1}^{i-1}\textrm{Tr}\left(\bar{\hat{\mathbf{\Sigma}}}_{l}\mathbf{W}_{l}\right)=\textrm{Tr}\left(\bar{\hat{\mathbf{\Sigma}}}_{i}\mathbf{W}_{i}\right),\nonumber \end{align}
where $\bar{\mathbf{\hat{\mathbf{\Omega}}}}_{i}=\mathbf{I}+\sum_{k=1}^{i-1}\mathbf{\Phi}_{k,i}\mathbf{H}_{k,i}^{\dagger}\mathbf{\bar{\hat{\mathbf{\Sigma}}}}_{k}\mathbf{H}_{k,i}$
and $\mathbf{W}_{i}=\mathbf{I}+\sum_{j=i+1}^{L}\mathbf{\Phi}_{i,j}\mathbf{H}_{i,j}\mathbf{\bar{\Sigma}}_{j}\mathbf{H}_{i,j}^{\dagger}=\bar{\mathbf{\Omega}}_{i}$.
}By a simple extension of the solution with\textcolor{black}{{} white
noise and sum power constraint in \cite{Telatar_EuroTrans_1999_MIMOCapacity}
to case }of colored noise and linear constraint here\textcolor{black}{,
it can be proved that the optimal solution is uniquely given by the
following polite water-filling procedure. Perform the thin SVD $\bar{\mathbf{\Omega}}_{i}^{-1/2}\mathbf{H}_{i,i}\bar{\mathbf{\hat{\mathbf{\Omega}}}}_{i}^{-1/2}=\mathbf{F}_{i}\mathbf{\Delta}_{i}\mathbf{G}_{i}^{\dagger}$.
Let $N_{i}=\textrm{Rank}\left(\mathbf{H}_{i,i}\right)$ and $\delta_{i,j}$
be the $j^{th}$ diagonal element of $\mathbf{\Delta}_{i}^{2}$. Obtain
$\mathbf{D}_{i}$ as \begin{eqnarray}
\mathbf{D}_{i} & = & \textrm{diag}\left(d_{i,j},...,d_{i,N_{i}}\right),\label{eq:WFDlS}\\
d_{i,j} & = & \left(\nu_{i}-\frac{1}{\delta_{i,j}}\right)^{+},j=1,...,N_{i},\nonumber \end{eqnarray}
where $\nu_{i}$ is chosen such that $\sum_{j=1}^{N_{l}}d_{i,j}=\textrm{Tr}\left(\bar{\hat{\mathbf{\Sigma}}}_{i}\bar{\mathbf{\Omega}}_{i}\right)$
and can be obtained by conventional water-filling algorithm. Then
the optimal solution is given by \begin{equation}
\hat{\mathbf{\Sigma}}_{i}^{'}=\bar{\mathbf{\Omega}}_{i}^{-1/2}\mathbf{F}_{i}\mathbf{D}_{i}\mathbf{F}_{i}^{\dagger}\bar{\mathbf{\Omega}}_{i}^{-1/2}\label{eq:optsighsub}\end{equation}
By Theorem \ref{thm:FequRGWF}, if $\mathbf{\bar{\Sigma}}_{i}$ does
not satisfy the polite water-filling, nor does $\bar{\hat{\mathbf{\Sigma}}}_{i}$,
which implies that $\bar{\hat{\mathbf{\Sigma}}}_{i}$ is not the optimal
solution and $\hat{\mathbf{\Sigma}}_{i}^{'}$ achieves a rate $\hat{\mathcal{I}}_{i}^{'}>\hat{\mathcal{I}}_{i}$. }

\textbf{Step} 2: Improve \textcolor{black}{the forward links} by the
covariance transformation from $\mathbf{\hat{\Sigma}}_{1:i}^{'}=\left(\bar{\hat{\mathbf{\Sigma}}}_{1},...,\bar{\hat{\mathbf{\Sigma}}}_{i-1},\hat{\mathbf{\Sigma}}_{i}^{'}\right)$
to $\mathbf{\Sigma}_{1:i}^{'}=\left(\mathbf{\Sigma}_{1}^{'},...,\mathbf{\Sigma}_{i}^{'}\right)$
for the sub-network.\textcolor{black}{{} }Due to the special interference
structure of iTree networks, the calculation of the transmit powers
of the covariance transformation can be simplified to be calculated
one by one as follows.\textcolor{black}{{} When calculating $p_{l,m}$,
the transmit powers $p_{k,m}:\ m=1,...,M_{k},\ k=l+1,...i$ and $p_{l,n},\ n=m+1,...,M_{l}$
have been calculated. }Therefore, we can calculate $\mathbf{\Sigma}_{k}^{'}=\sum_{m=1}^{M_{k}}p_{k,m}\mathbf{t}_{k,m}\mathbf{t}_{k,m}^{\dagger},\ k=l+1,...i$
and obtain the interference-plus-noise covariance matrix of the link
$l$ as $\mathbf{\Omega}_{l}^{'}=\mathbf{I}+\sum_{k=l+1}^{i}\mathbf{\Phi}_{l,k}\mathbf{H}_{l,k}\mathbf{\Sigma}_{k}^{'}\mathbf{H}_{l,k}^{\dagger}$.
Then \textcolor{black}{obtain $p_{l,m}$ as\begin{equation}
p_{l,m}=\frac{\hat{\gamma}_{l,m}\left(\mathbf{r}_{l,m}^{\dagger}\mathbf{\Omega}_{l}^{'}\mathbf{r}_{l,m}+\sum_{n=m+1}^{M_{l}}p_{l,n}\left|\mathbf{r}_{l,m}^{\dagger}\mathbf{H}_{l,l}\mathbf{t}_{l,n}\right|^{2}\right)}{\left|\mathbf{r}_{l,m}^{\dagger}\mathbf{H}_{l,l}\mathbf{t}_{l,m}\right|^{2}}.\label{eq:FLpowtree}\end{equation}
Finally $\mathbf{\Sigma}_{l}^{'}$ is given by\begin{eqnarray}
\mathbf{\Sigma}_{l}^{'} & = & \sum_{m=1}^{M_{l}}p_{l,m}\mathbf{t}_{l,m}\mathbf{t}_{l,m}^{\dagger}.\label{eq:Sigmatree}\end{eqnarray}
}By Theorem \ref{thm:linear-color-dual}\textcolor{black}{, the covariance
transformation $\mathbf{\Sigma}_{1:i}^{'}$ achieves a set of rates
$\mathcal{I}_{i}^{'}\geq\hat{\mathcal{I}}_{i}^{'}>\mathcal{I}_{i}$
and $\mathcal{I}_{l}^{'}\geq\hat{\mathcal{I}}_{l}^{'}=\mathcal{I}_{l},l<i$
 in the sub-network under the sum power constraint $\sum_{l=1}^{i}\textrm{Tr}\left(\mathbf{\Sigma}_{l}^{'}\right)=\sum_{l=1}^{i-1}\textrm{Tr}\left(\bar{\hat{\mathbf{\Sigma}}}_{l}\mathbf{W}_{l}\right)+\textrm{Tr}\left(\hat{\mathbf{\Sigma}}_{i}^{'}\mathbf{W}_{i}\right)=P_{T}^{i}$.
Noting that the first $i$ links cause no interference to all other
links in the original network, the input covariance matrices $\mathbf{\Sigma}_{1:L}^{'}=\left(\mathbf{\Sigma}_{1}^{'},\cdots,\mathbf{\Sigma}_{i}^{'},\mathbf{\bar{\Sigma}}_{i+1},\cdots,\mathbf{\bar{\Sigma}}_{L}\right)$
must achieve a rate point $\mathcal{I}_{i}^{'}>\mathcal{I}_{i}$ and
$\mathcal{I}_{l}^{'}\geq\mathcal{I}_{l},\forall l\neq i$ in the original
network with the same sum power $P_{T}$. }

\textcolor{black}{We refer to the above algorithm as Algorithm S and
summarize it in table \ref{tab:table4}.}

\begin{table}
\caption{\label{tab:table4}Algorithm S (Improving the Rate of Link $i$ for
iTree Networks)}

\centering{}\begin{tabular}{l}
\hline
\noalign{\vskip\doublerulesep}
{\small 1. Obtain $\bar{\mathbf{\hat{\mathbf{\Sigma}}}}_{1:i}$ from
the solution of Algorithm A or B.}\tabularnewline
{\small 2. Solve for the optimal $\mathbf{\hat{\mathbf{\Sigma}}}_{i}^{'}$
in the optimization problem (\ref{eq:MRSC})}\tabularnewline
{\small \ \ \ by polite water-filling.}\tabularnewline
{\small 3. Calculate the input covariance matrices }\textcolor{black}{\small $\mathbf{\Sigma}_{l}^{'},l=1,...,i$
}{\small for the first}\tabularnewline
{\small \ \ \ $i$ links by the covariance transformation as in
(\ref{eq:Sigmatree}). }\tabularnewline
{\small 4. }\textbf{\small Output}{\small{} the updated input covariance
matrices }\tabularnewline
{\small \ \ \ \ \ \ \ \ }\textcolor{black}{\small $\mathbf{\Sigma}_{1:L}^{'}=\left(\mathbf{\Sigma}_{1}^{'},\cdots,\mathbf{\Sigma}_{i}^{'},\mathbf{\bar{\Sigma}}_{i+1},\cdots,\mathbf{\bar{\Sigma}}_{L}\right).$ }\tabularnewline
\end{tabular}
\end{table}

\textcolor{black}{The performance can be strictly improved using the
output of Algorithm S. For }\textbf{\textcolor{black}{FOP}}\textcolor{black}{,
we first reduce the transmit power of link $i$ until its rate is
reduced to $\mathcal{I}_{i}$. This may benefit other links as well
because the interference to other links is also reduced. Then this
extra power can be used to simultaneously increase all link's power
by the same factor, and thus improve the rates of all links and the
objective function of }\textbf{\textcolor{black}{FOP}}\textcolor{black}{.
The cost function of }\textbf{\textcolor{black}{SPMP}}\textcolor{black}{{}
can be strictly decreased by reducing the transmit power of link $i$
such that the rate is reduced to $\mathcal{I}_{i}^{0}$. Note that
the above operations can be automatically achieved by Algorithm A
(B) using $\mathbf{\Sigma}_{1:L}^{'}$ as the initial point. }According
to the above, we propose an improved algorithm for iTree networks
using Algorithm S as a component. It is referred to as Algorithm I
and summarized in table \ref{tab:table3}.\textcolor{black}{{} The optimality
of Algorithm I is stated in the theorem below.}

\begin{table}
\caption{\label{tab:table3}Algorithm I (Improved Algorithm for iTree Networks)}

\centering{}\begin{tabular}{l}
\hline
{\small 1. Generate random initial point such that all streams have}\tabularnewline
{\small \ \ \ positive SINR.}\tabularnewline
{\small 2. Run Algorithm A or B to obtain the solution $\bar{\mathbf{\Sigma}}_{1:L}$.}\tabularnewline
{\small 3. }\textbf{\small Repeat}\tabularnewline
{\small \ \ \ \ }\textbf{\small For}{\small{} $i=1:L$ }\tabularnewline
{\small \ \ \ \ \ \ }\textbf{\small If}{\small{} $\bar{\mathbf{\Sigma}}_{i}$
does not satisfy the polite water-filling structure }\tabularnewline
{\small \ \ \ \ \ \ \ \ Obtain $\mathbf{\Sigma}_{1:L}^{'}$
from $\bar{\mathbf{\Sigma}}_{1:L}$ using Algorithm S.}\tabularnewline
{\small \ \ \ \ \ \ \ \ Run Algorithm A or B with $\mathbf{\Sigma}_{1:L}^{'}$
as the initial point to}\tabularnewline
{\small \ \ \ \ \ \ \ \ obtain the solution $\bar{\mathbf{\Sigma}}_{1:L}$.}\tabularnewline
{\small \ \ \ \ \ \ }\textbf{\small End}{\small{} }\tabularnewline
{\small \ \ \ \ }\textbf{\small End}\tabularnewline
{\small \ \ }\textbf{\small Until}{\small{} converge}\tabularnewline
\end{tabular}
\end{table}

\begin{thm}
\textcolor{black}{\label{thm:opttreenet}For iTree networks, the solution
of Algorithm I satisfies all the optimality conditions in Theorem
\ref{thm:optimality}, and thus achieves a stationary point.}\end{thm}
\begin{proof}
For iTree networks, the \textbf{SPMP} is always feasible. Step 2 of
the Algorithm I generates a feasible solution $\bar{\mathbf{\Sigma}}_{1:L}$.
Then in step 3, both Algorithm S and Algorithm A (B) will monotonically
increase (decrease) the objective (cost) in \textbf{FOP} (\textbf{SPMP}).
Since the objective (cost) is upper bounded (lower bounded), Algorithm
I must converge to a fixed point. If \textcolor{black}{the optimality
conditions in Theorem \ref{thm:optimality}} is not satisfied, Algorithm
I will strictly increase (decrease) the objective (cost) in \textbf{FOP}
(\textbf{SPMP}), which contradicts with the assumption of fixed point.
\end{proof}

\textcolor{black}{Actually, in almost all simulations we conducted
for general B-MAC networks, Algorithm A or B with initial point randomly
generated from a continuous space is observed to converge to a stationary
point. Therefore, Algorithm I seldomly runs Algorithm S. Therefore,
Algorithm S is more of theoretic value for the convergence to a stationary
point and serve as a basis for the algorithm in the next subsection.}

\subsection{\label{sub:PR-PR1}Polite Water-filling based Algorithms for \textup{SPMP}}

\subsubsection{Algorithm PR for iTree Networks}

In stead of converting the rate constrained problem to the SINR constrained
problem, we can modify Algorithm S to directly solve \textbf{SPMP}
for iTree networks. In (\ref{eq:WFDlS}), the polite water-filling
level \textcolor{black}{$\nu_{i}$ is chosen such that the sum power
is unchanged when switching to the forward links. Because this polite
water-filling level will improve the rate of reverse link $i$, if
the initial solution is feasible, i.e., the rate of reverse link $i$
is no less than $\mathcal{I}_{i}^{0}$, we can reduce the polite water-filling
level $\nu_{i}$ to make the rate of reverse link $i$ equal to $\mathcal{I}_{i}^{0}$,
and thus reduce the sum power when switching to the forward links.
This results in an algorithm which monotonically decreases the sum
power once the solution becomes feasible. A simple algorithm in Table
\ref{tab:alg-W} referred to as }\textit{\textcolor{black}{Algorithm
W}}\textcolor{black}{{} can be used to calculate the polite water-filling
level $\nu_{i}$} to satisfy the rate constraint $\mathcal{I}_{i}^{0}$.

\begin{table}
\caption{\label{tab:alg-W}\textit{\emph{Algorithm W}} (Solving the Polite
Water-filling Level for the Rate Constraints)}

\centering{}\begin{tabular}{l}
\hline
{\small 1. Initialize the set of indices of the streams of link $i$
as}\tabularnewline
{\small \ \ \ $\Gamma=\left\{ 1,...,N_{i}\right\} $, where }\textcolor{black}{\small $N_{i}=\textrm{Rank}\left(\mathbf{H}_{i,i}\right)$}{\small .}\tabularnewline
{\small 2. Calculate $\nu_{i}=\left(e^{\mathcal{I}_{i}^{0}}/\Pi_{j\in\Gamma}\delta_{i,j}\right)^{1/\left|\Gamma\right|}$,
which is the solution of }\tabularnewline
{\small \ \ \ $\sum_{j\in\Gamma}\textrm{log}\left(1+\left(\nu_{i}-1/\delta_{i,j}\right)\delta_{i,j}\right)=\mathcal{I}_{i}^{0}$. }\tabularnewline
{\small \ \ \ Obtain $d_{i,j}=\nu_{i}-1/\delta_{i,j}$ for $j\in\Gamma$.}\tabularnewline
{\small 3. If $d_{i,j}\ge0$, $\forall j\in\Gamma$, stop. Otherwise,
for all $j\in\Gamma$, if $d_{i,j}<0$, }\tabularnewline
{\small \ \ \ fix it as $d_{i,j}=0$, delete $j$ from $\Gamma$.
Repeat step 2).}\tabularnewline
\end{tabular}
\end{table}

This modification of Algorithm S is referred to as Algorithm PR and
is summarized in table \ref{tab:table7}, where P stands for Polite
and R stands for Rate constraint. It can be shown that once Algorithm
PR finds a feasible solution, it will monotonically decrease the sum
power until it converges to a stationary point.

\begin{table}
\caption{\label{tab:table7}Algorithm PR (Solving\textbf{ SPMP} for iTree Networks)}

\centering{}\begin{tabular}{l}
\hline
{\small Initialize $\mathbf{\Sigma}_{1:L}$ such that $\mathbf{\Sigma}_{i}\succeq0,\forall i$.}\tabularnewline
\textbf{\small While}{\small{} not converge }\textbf{\small do}{\small{} }\tabularnewline
{\small $\;$}\textbf{\small For}{\small{} $i=1:L$}\tabularnewline
{\small $\;$$\;$1. Calculate $\hat{\mathbf{\Sigma}}_{1:i}$ by the
covariance transformation of $\mathbf{\Sigma}_{1:i}$}\tabularnewline
{\small $\;$$\;$$\;$$\;$$\;$applied to the $i^{\text{th}}$ sub-network.}\tabularnewline
{\small $\;$$\;$2. Obtain $\mathbf{\hat{\mathbf{\Sigma}}}_{i}^{'}$
by polite water-filling as in (\ref{eq:WFDlS}) and (\ref{eq:optsighsub}),
where}\tabularnewline
{\small $\;$$\;$$\;$$\;$$\;$the polite water-filling level $\nu_{i}$
is calculated by Algorithm W.}\tabularnewline
{\small $\;$$\;$3. Calculate $\mathbf{\Sigma}_{1:i}^{'}$ by the
covariance transformation of }\tabularnewline
{\small $\;$$\;$$\;$$\;$$\;$$\mathbf{\hat{\mathbf{\Sigma}}}_{1:i}^{'}=\left(\mathbf{\hat{\Sigma}}_{1},...,\mathbf{\hat{\Sigma}}_{i-1},\mathbf{\hat{\Sigma}}_{i}^{'}\right)$
applied to the $i^{\text{th}}$ sub-network.}\tabularnewline
{\small $\;$$\;$4. Update $\mathbf{\Sigma}_{1:L}$ as $\mathbf{\Sigma}_{1:L}=\left(\mathbf{\Sigma}_{1}^{'},...,\mathbf{\Sigma}_{i}^{'},\mathbf{\Sigma}_{i+1},...,\mathbf{\Sigma}_{L}\right)$.}\tabularnewline
{\small $\;$}\textbf{\small End}\tabularnewline
\textbf{\small End}\tabularnewline
\end{tabular}
\end{table}

\subsubsection{Algorithm PR1 for B-MAC Networks}

To get rid of the covariance transformation which is more complex
than the polite water-filling as will be discussed later, and to make
the algorithm work for general B-MAC networks, we obtain an intuitive
algorithm by imposing the polite water-filling structure iteratively.
It is referred to as Algorithm PR1 and is summarized in table \ref{tab:table6}.
It turns out that the algorithm can also be derived from the Lagrange
function and KKT conditions of the problem, where the Lagrange multipliers
are exactly the water-filling levels of the links. Adjusting the Lagrange
multipliers to satisfy the rate constraints is exactly what Algorithm
W does. It is clear that if the algorithm converges, the solution
of Algorithm PR1 satisfies the optimality conditions in Theorem \ref{thm:optimality},
and thus achieves a stationary point.

\begin{table}
\caption{\label{tab:table6}Algorithm PR1 (Solving\textbf{ SPMP} for B-MAC
Networks)}

\centering{}\begin{tabular}{l}
\hline
{\small Initialize $\mathbf{\hat{\Sigma}}_{1:L}$ and $\mathbf{\Omega}_{i}$'s
such that $\mathbf{\hat{\Sigma}}_{i}\succeq0,\forall i$ and $\mathbf{\Omega}_{i}=\mathbf{I},\forall i$.}\tabularnewline
\textbf{\small While}{\small{} not converge }\textbf{\small do}{\small{} }\tabularnewline
{\small $\;$1. Update in the forward links }\tabularnewline
{\small $\;$$\;$a. For $\forall i$, obtain $\hat{\mathbf{\Omega}}_{i}$
from $\hat{\mathbf{\Sigma}}_{1:L}$ using (\ref{eq:WhiteMRV}).}\tabularnewline
{\small $\;$$\;$$\;$$\;$$\;$$\;$Perform thin SVD $\mathbf{\Omega}_{i}^{-1/2}\mathbf{H}_{i,i}\hat{\mathbf{\Omega}}_{i}^{-1/2}=\mathbf{F}_{i}\mathbf{\Delta}_{i}\mathbf{G}_{i}^{\dagger}$. }\tabularnewline
{\small $\;$$\;$b. Obtain $\mathbf{D}_{i}$ by the water-filling
in (\ref{eq:WFDlS}), where}\tabularnewline
{\small $\;$$\;$$\;$$\;$$\;$$\;$the polite water-filling level
$\nu_{i}$ is calculated by Algorithm W.}\tabularnewline
{\small $\;$$\;$c. Update $\mathbf{\Sigma}_{i}$'s as}\tabularnewline
{\small $\;$$\;$$\;$$\;$$\;$$\;$$\;$$\;$$\mathbf{\Sigma}_{i}=\hat{\mathbf{\Omega}}_{i}^{-1/2}\mathbf{G}_{i}\mathbf{D}_{i}\mathbf{G}_{i}^{\dagger}\hat{\mathbf{\Omega}}_{i}^{-1/2},\forall i$.}\tabularnewline
{\small $\;$2. Update in the reverse links }\tabularnewline
{\small $\;\;$a. For $\forall i$, obtain $\mathbf{\Omega}_{i}$
from $\mathbf{\Sigma}_{1:L}$ using (\ref{eq:whiteMG}).}\tabularnewline
{\small $\;$$\;$$\;$$\;$$\;$$\;$Perform thin SVD $\mathbf{\Omega}_{i}^{-1/2}\mathbf{H}_{i,i}\hat{\mathbf{\Omega}}_{i}^{-1/2}=\mathbf{F}_{i}\mathbf{\Delta}_{i}\mathbf{G}_{i}^{\dagger}$. }\tabularnewline
{\small $\;$$\;$b. Obtain $\mathbf{D}_{i}$ by the water-filling
in (\ref{eq:WFDlS}), where}\tabularnewline
{\small $\;$$\;$$\;$$\;$$\;$$\;$the polite water-filling level
$\nu_{i}$ is calculated by Algorithm W.}\tabularnewline
{\small $\;$$\;$c. Update $\mathbf{\hat{\Sigma}}_{i}$'s as}\tabularnewline
{\small $\;$$\;$$\;$$\;$$\;$$\;$$\;$$\;$$\mathbf{\hat{\Sigma}}_{i}=\mathbf{\Omega}_{i}^{-1/2}\mathbf{F}_{i}\mathbf{D}_{i}\mathbf{F}_{i}^{\dagger}\mathbf{\Omega}_{i}^{-1/2},\forall i$.}\tabularnewline
\textbf{\small End}\tabularnewline
\end{tabular}
\end{table}

It is difficult to prove the convergence of Algorithm PR1. But the
intuition and all simulations we conducted strongly indicate fast
convergence.

\begin{remrk}
Algorithm PR1 can be used to solve the \textbf{FOP} by replacing constraints
$\mathcal{I}_{l}^{0}$ with $\alpha\mathcal{I}_{l}^{0}$ and searching
for $\alpha$ to satisfy the power constraint.
\end{remrk}

\begin{remrk}
An advantage of Algorithm PR1 is that it can be easily implemented
distributedly as will be shown in Section \ref{sub:Distributed-Implementation}.
Another advantage is that it has linear complexity in each iteration,
because the SVD for polite water-filling is performed over the matrices
whose dimensions are not increased with the number of desired data
links $L$. However, the complexity order of Algorithm B depends on
$L$. In each iteration, the complexity of calculating all the MMSE-SIC
receive vectors is still linear with respect to $L$. But to calculate
the transmit powers $\mathbf{p}$ and $\mathbf{q}$, we need to solve
two $\sum_{l=1}^{L}M_{l}$-dimensional linear equations, whose complexity
depends on the density and structure of the cross-talk matrix $\mathbf{\Psi}\left(\mathbf{T},\mathbf{R}\right)$.
In the worst case, the complexity order is $\mathcal{O}\left(L^{3}\right)$.
In other cases such as with triangular or sparse $\mathbf{\Psi}\left(\mathbf{T},\mathbf{R}\right)$,
the complexity is much lower. Fortunately, in practice, $\mathbf{\Psi}\left(\mathbf{T},\mathbf{R}\right)$
is usually sparse for a large wireless network because of path loss.%
{}
\end{remrk}

\subsection{\label{sub:Distributed-Implementation}Distributed Implementation
of Algorithm PR1}

\textcolor{black}{In a network, it is desirable to use distributed
optimization. The above centralized algorithms serve as the basis
for distributed design. Here, we design a distributed algorithm based
on Algorithm PR1 for time division duplex (TDD) networks.. To perform
the polite water-filling in Algorithm PR1, $T_{l}$ ($R_{l}$) only
needs to know the equivalent channel $\mathbf{\Omega}_{l}^{-1/2}\mathbf{H}_{l,l}\hat{\mathbf{\Omega}}_{l}^{-1/2}$,
where both $\mathbf{\hat{\Omega}}_{l}$ ($\mathbf{\Omega}_{l}$) and
$\mathbf{H}_{l,l}^{\dagger}\mathbf{\Omega}_{l}^{-1/2}$ ($\mathbf{H}_{l,l}\hat{\mathbf{\Omega}}_{l}^{-1/2}$)
can be obtained by pilot-aided estimation in the reverse (forward)
links in time division duplex (TDD) networks. In frequency division
duplex (FDD) networks, the equivalent channel needs to be calculated
from feedback. Thus, TDD system has an advantage.}

We assume block fading channel, where each block consists of a training
stage followed by a transmission stage. The training stage is further
divided into several rounds, where one round consists of a half round
of pilot aided estimation of $\mathbf{H}_{l,l}\hat{\mathbf{\Omega}}_{l}^{-1/2}$
and $\mathbf{\Omega}_{l}$ in the forward link and a half round of
pilot aided estimation of $\mathbf{H}_{l,l}^{\dagger}\mathbf{\Omega}_{l}^{-1/2}$
and $\mathbf{\hat{\Omega}}_{l}$ in the reverse link. Since Algorithm
PR1 achieves most of the benefit in very few iterations, the required
number of training rounds is small and can be as less as 2.5 rounds.
In the training stage, the $T_{l}$'s and $R_{l}$'s run a distributed
version of \textcolor{black}{Algorithm PR1} to solve \textbf{SPMP}.
At the end of the training stage, $T_{l}$'s use the latest $\mathbf{\Sigma}_{l}$'s
as the input covariance matrices for the transmission stage.

First, we describe the operation at each node.

\textit{Operation at }$T_{l}$:
\begin{itemize}
\item In the $\left(i-1\right)^{\text{th}}$ reverse \textcolor{black}{training}
round, \textcolor{black}{$T_{l}$ estimates the interference-plus-noise
covariance matrix $\mathbf{\hat{\Omega}}_{l}^{(i-1)}$ and the }effective\textcolor{black}{{}
channel $\mathbf{H}_{l,l}^{\dagger}\left(\mathbf{\Omega}_{l}^{(i-1)}\right)^{-1/2}$. }
\item \textcolor{black}{At the beginning of the $i^{\text{th}}$ forward
training round, $T_{l}$ calculates the input covariance matrix $\mathbf{\Sigma}_{l}^{(i)}$
by polite water-filling over the equivalent channel $\left(\mathbf{\Omega}_{l}^{(i-1)}\right)^{-1/2}\mathbf{H}_{l,l}\left(\mathbf{\hat{\Omega}}_{l}^{(i-1)}\right)^{-1/2}$
as in step 1 of Algorithm PR1. However, in practice, the transmit
power of $T_{l}$ can not exceed certain maximum value denoted by
$P_{T_{l}}^{\textrm{max}}$. If $\textrm{Tr}\left(\mathbf{\Sigma}_{l}^{(i)}\right)>P_{T_{l}}^{\textrm{max}}$,
}decrease the polite water-filling level $\nu_{l}$ obtained by Algorithm
W until \textcolor{black}{$\textrm{Tr}\left(\mathbf{\Sigma}_{l}^{(i)}\right)=P_{T_{l}}^{\textrm{max}}$}.\textcolor{black}{{}
Then $T_{l}$ transmits pilot signals in the $i^{\text{th}}$ forward
training round. }
\end{itemize}

\textit{Operation at }$R_{l}$:
\begin{itemize}
\item In the $i^{\text{th}}$ forward \textcolor{black}{training} round,
\textcolor{black}{$R_{l}$ estimates $\mathbf{\Omega}_{l}^{(i)}$
and $\mathbf{H}_{l,l}\left(\mathbf{\hat{\Omega}}_{l}^{(i-1)}\right)^{-1/2}$. }
\item \textcolor{black}{At the beginning of the $i^{th}$ reverse training
round, $R_{l}$ calculates $\hat{\mathbf{\Sigma}}_{l}^{(i)}$ by polite
water-filling over the equivalent channel $\left(\mathbf{\hat{\Omega}}_{l}^{(i-1)}\right)^{-1/2}\mathbf{H}_{l,l}^{\dagger}\left(\mathbf{\Omega}_{l}^{(i)}\right)^{1/2}$
as in step 2 of Algorithm PR1. Denote $P_{R_{l}}^{\textrm{max}}$
as the maximum transmit power of $R_{l}$. If $\textrm{Tr}\left(\hat{\mathbf{\Sigma}}_{l}^{(i)}\right)>P_{R_{l}}^{\textrm{max}}$,
}decrease the polite water-filling level $\nu_{l}$ obtained by Algorithm
W until \textcolor{black}{$\textrm{Tr}\left(\hat{\mathbf{\Sigma}}_{l}^{(i)}\right)=P_{R_{l}}^{\textrm{max}}$}.\textcolor{black}{{}
Then $R_{l}$ transmits pilot signals in the $i^{\text{th}}$ reverse
training round.}
\end{itemize}

Then, we summarize this distributed algorithm for calculating the
input covariance matrices in Table \textcolor{black}{\ref{tab:table9}
and refer to it as }\textit{\textcolor{black}{Algorithm PRD}}\textcolor{black}{.}

\begin{table}
\caption{\label{tab:table9}Algorithm PRD (Distributed Version of Algorithm
PR1) }

\centering{}\begin{tabular}{l}
\hline
{\small Initialize}\textcolor{black}{\small{} $i=1$ and}{\small{} $\mathbf{\Omega}_{l}^{(0)}=\mathbf{I},\:\mathbf{\hat{\Omega}}_{l}^{(0)}=\mathbf{I},\:\forall l$.}\tabularnewline
{\small 1. }\textcolor{black}{\small In the $i^{\textrm{th}}$ forward
training round, $T_{l}$ calculates $\mathbf{\Sigma}_{l}^{(i)}$ by
polite }\tabularnewline
{\small $\;\;$$\;\;$}\textcolor{black}{\small water-filling over
$\left(\mathbf{\Omega}_{l}^{(i-1)}\right)^{-1/2}\mathbf{H}_{l,l}\left(\mathbf{\hat{\Omega}}_{l}^{(i-1)}\right)^{-1/2}$
and transmits}\tabularnewline
{\small $\;\;$$\;\;$}\textcolor{black}{\small pilot signals. (In
the $1^{\text{st}}$ round, since the equivalent channel is}\tabularnewline
{\small $\;\;$$\;\;$}\textcolor{black}{\small unknown, $\mathbf{\Sigma}_{l}^{(1)}$
can be chosen randomly.) }\tabularnewline
{\small $\;\;$$\;\;$}\textcolor{black}{\small $R_{l}$ estimates
$\mathbf{\Omega}_{l}^{(i)}$ and $\mathbf{H}_{l,l}\left(\mathbf{\hat{\Omega}}_{l}^{(i-1)}\right)^{-1/2}$.}\tabularnewline
{\small 2. }\textcolor{black}{\small In the $i^{\textrm{th}}$ reverse
training round, $R_{l}$ calculates $\mathbf{\hat{\Sigma}}_{l}^{(i)}$
by polite}\tabularnewline
{\small $\;\;$$\;\;$}\textcolor{black}{\small water-filling over
$\left(\mathbf{\hat{\Omega}}_{l}^{(i-1)}\right)^{-1/2}\mathbf{H}_{l,l}^{\dagger}\left(\mathbf{\Omega}_{l}^{(i)}\right)^{-1/2}$
and transmits}\tabularnewline
{\small $\;\;$$\;\;$}\textcolor{black}{\small pilot signals. $T_{l}$
estimates $\mathbf{\hat{\Omega}}_{l}^{(i)}$ and $\mathbf{H}_{l,l}^{\dagger}\left(\mathbf{\Omega}_{l}^{(i)}\right)^{-1/2}$.}\tabularnewline
{\small 3.}\textcolor{black}{\small{} Let $i=i+1$ and enter the next
round. Keep updating $\mathbf{\Sigma}_{l}^{(i)}$ and }\tabularnewline
{\small $\;\;$$\;\;$}\textcolor{black}{\small $\hat{\mathbf{\Sigma}}_{l}^{(i)}$
until the end of the training stage.}\tabularnewline
\end{tabular}
\end{table}

Note that one training round is almost the same as one iteration in
Algorithm PR1 except that the transmit power of each node is constrained
not to exceed a maximum value. Therefore, after convergence, if the
transmit power of each node does not exceed its maximum transmit power\textcolor{black}{,
Algorithm PRD achieves nearly the same performance as Algorithm PR1.
However in practice, it is desirable to allocate as less number of
training rounds as possible. Then, Algorithm} \textcolor{black}{PRD
may not be able to fully converge and thus the target rates may not
be satisfied. Fortunately, it is observed in the simulations that
Algorithm PRD converges very fast at the first few rounds.}

\subsection{\label{sub:OrderOptimization}Optimization of the Encoding and Decoding
Order}

For the special case of using DPC and SIC to cancel interference,
the coupling matrix $\mathbf{\Phi}\left(\pi\right)$ is a function
of the encoding and decoding order $\pi$. In this section, we partially
characterize the optimal encoding and decoding order $\pi$ and design
an algorithm to find a good $\pi$. It is in general a difficult problem
because the encoding and decoding orders at the BC transmitters and
the MAC receivers need to be solved jointly. However, for \textcolor{black}{each
}\textit{\textcolor{black}{Pseudo }}\textit{BC}/\textit{\textcolor{black}{Pseudo
}}\textit{MAC} defined below, the optimal $\pi$ is characterized
in Theorem \ref{thm:optorder}.
\begin{definitn}
In a B-MAC network, a sub-network with a single physical transmitter
and a set of associated links whose indices forms a set $\mathcal{L}_{\textrm{B}}$
is said to be a \textit{\textcolor{black}{Pseudo }}\textit{BC} if
either all links in $\mathcal{L}_{\textrm{B}}$ completely interfere
with a link $k$ or all links in $\mathcal{L}_{\textrm{B}}$ do not
interfere with a link $k$, $\forall k\in\mathcal{L}_{\textrm{B}}^{C}$,
i.e., the columns with indices in $\mathcal{L}_{\textrm{B}}$ of the
coupling matrix $\mathbf{\Phi}$, excluding rows in $\mathcal{L}_{\textrm{B}}$,
are the same. A sub-network with a single physical receiver and a
set of associated links, whose indices forms a set $\mathcal{L}_{\textrm{M}}$,
is said to be a \textit{\textcolor{black}{Pseudo }}\textit{MAC} if
either all links in $\mathcal{L}_{\textrm{M}}$ are completely interfered
by a link $k$ or all links in $\mathcal{L}_{\textrm{M}}$ are not
interfered by a link $k$, $\forall k\in\mathcal{L}_{\textrm{M}}^{C}$,
i.e., the rows with indices in $\mathcal{L}_{\textrm{M}}$ of the
coupling matrix $\mathbf{\Phi}$, excluding columns in $\mathcal{L}_{\textrm{M}}$,
are the same.
\end{definitn}

\begin{example}
\label{exa:pseudoBM}In Fig. \ref{fig:sysFig1}, suppose $\mathbf{x}_{1}$
is encoded after $\mathbf{x}_{2}$ and $\mathbf{x}_{4}$ is the last
one to be decoded at the second physical receiver. Then link 2 and
link 3 forms a pseudo MAC because they belong to the same physical
receiver and suffer the same interference from $\mathbf{x}_{1}$,
$\mathbf{x}_{4}$ and $\mathbf{x}_{5}$. Similarly, link 4 and link
5 forms a pseudo BC. \end{example}
\begin{remrk}
The \textcolor{black}{pseudo }BC and \textcolor{black}{pseudo }MAC
were first introduced in \cite{Liu_IT10s_Duality_BMAC} where the
optimal encoding/decoding order of \textcolor{black}{each pseudo }BC/\textcolor{black}{pseudo
}MAC for the weighted sum-rate maximization problem (\textbf{WSRMP})
is shown to be consistent with the optimal order of an individual
BC or MAC. This is because \textcolor{black}{a Pseudo }BC\textcolor{black}{{}
or a Pseudo }MAC can be isolated from the B-MAC network to form an
individual BC or MAC. Similar results can also be obtained for \textbf{FOP}
and \textbf{SPMP} as shown below.
\end{remrk}

First, we need to modify the \textbf{FOP} and \textbf{SPMP} to include
encoding and decoding order optimization and time sharing. Let $\mathbf{\Xi}$
be a set of valid coupling matrices produced by proper encoding and
decoding orders of a B-MAC network. Define a larger achievable region
\begin{eqnarray*}
\mathcal{R}(P_{T}) & = & \text{Convex Closure}\bigcup_{\mathbf{\Phi}\in\mathbf{\Xi}}\mathcal{R}_{\mathbf{\Phi}}\left(P_{T}\right).\end{eqnarray*}
 The modified optimization problems are

\textcolor{black}{\begin{eqnarray}
\textrm{\textbf{OFOP}:} & \underset{\mathbf{r}\in\mathcal{R}(P_{T})}{\textrm{max}} & \left(\min_{1\leq l\leq L}\frac{r_{l}}{\mathcal{I}_{l}^{0}}\right),\label{eq:P1-timesharing}\end{eqnarray}
 and \begin{eqnarray}
\textrm{\textbf{OSPMP}:} & \underset{P_{T}}{\textrm{min}} & P_{T}\label{eq:P2-timesharing}\\
 & \textrm{s.t.} & [\mathcal{I}_{l}^{0}]_{l=1,...,L}\in\mathcal{R}(P_{T}).\nonumber \end{eqnarray}
}

The following lemma is a consequence of that all outer boundary points
of $\mathcal{R}(P_{T})$ are Pareto optimal and can be proved by contradiction.
\begin{lemma}
\label{lem:optodr}The\textcolor{black}{{} optimal solution of }\textbf{\textcolor{black}{OFOP}}\textcolor{black}{{}
or }\textbf{\textcolor{black}{OSPMP}}\textcolor{black}{{} is the intersection
of the ray $\alpha\left[\mathcal{I}_{1}^{0},\cdots,\mathcal{I}_{L}^{0}\right],\ \alpha>0$,
and the boundary of }$\mathcal{R}(P_{T})$\textcolor{black}{, where
for }\textbf{\textcolor{black}{OSPMP}}\textcolor{black}{, the sum
power $P_{T}$ is chosen such that the intersection is at $\left[\mathcal{I}_{1}^{0},\cdots,\mathcal{I}_{L}^{0}\right]$.}
\end{lemma}

\textcolor{black}{The following theorem }characterizes\textcolor{black}{{}
the optimal }encoding and decoding order\textcolor{black}{{} for those
boundary points of }$\mathcal{R}(P_{T})$\textcolor{black}{{} that can
be achieved without time sharing and by DPC and SIC. The sufficiency
part only holds for MAC and BC.}
\begin{thm}
\textcolor{black}{\label{thm:optorder}}Necessity: If the input covariance
matrices $\mathbf{\tilde{\Sigma}}_{1:L}$ and\textcolor{black}{{} the
encoding and decoding order} $\tilde{\pi}$ for a valid coupling matrix
achieves the\textcolor{black}{{} optimum of }\textbf{OFOP} and \textbf{OSPMP}\textcolor{black}{{}
without time sharing, they must satisfy the following }necessary\textcolor{black}{{}
conditions:}
\begin{enumerate}
\item $\mathbf{\tilde{\Sigma}}_{1:L}$\textcolor{black}{{} satisfies the optimality
conditions in Theorem \ref{thm:optimality}. }
\item If there exists a pseudo BC (pseudo MAC) in the B-MAC network, its
optimal encoding (decoding) order $\tilde{\pi}$ satisfies that, the
signal of the link with the $n^{\text{th}}$ largest \textcolor{black}{(}smallest\textcolor{black}{)
polite water-filling level} is the $n^{\text{th}}$ one to be encoded
(decoded).
\end{enumerate}
Sufficiency: In MAC or BC, if \textcolor{black}{certain $\mathbf{\tilde{\Sigma}}_{1:L}$
and }$\tilde{\pi}$\textcolor{black}{{} satisfy the above conditions,
then $\tilde{\mathbf{\Sigma}}_{1:L}$ and }$\tilde{\pi}$\textcolor{black}{{}
achieves the global optimum of }\textbf{\textcolor{black}{OFOP}}\textcolor{black}{{}
or }\textbf{\textcolor{black}{OSPMP}}\textcolor{black}{.}
\end{thm}

\begin{proof}
The first necessary condition follows %
{} from Theorem \textcolor{black}{\ref{thm:optimality}}. Then \textcolor{black}{the
second }necessary\textcolor{black}{{} condition} follows from the following
two facts and Lemma \ref{lem:optodr}. 1) Any outer boundary point
of $\mathcal{R}(P_{T})$ must be the solution of a \textbf{WSRMP}
with certain weight vector $\left[w_{l}\right]_{,l=1,...,L}$. It
is proved in \cite{Liu_IT10s_Duality_BMAC} \textcolor{black}{that
the optimal input covariance matrices }maximizing\textcolor{black}{{}
}the weighted sum-rate\textcolor{black}{{} must satisfy the polite water-filling
structure and the polite water-filling levels are given by $\nu w_{l}$'s
for some constant $\nu>0$; 2) By Theorem 9 in }\cite{Liu_IT10s_Duality_BMAC},
the weighted sum-rate optimal encoding and decoding order of each
\textcolor{black}{Pseudo }BC (\textcolor{black}{Pseudo }MAC) is that
the signal of the link with the $n^{\text{th}}$ largest (smallest)
\textcolor{black}{weight} is the $n^{\text{th}}$ one to be encoded
(decoded). For the sufficiency part, suppose\textcolor{black}{{} certain
$\mathbf{\tilde{\Sigma}}_{1:L}$ and }$\tilde{\pi}$ satisfy the two
conditions in \textcolor{black}{Theorem }\ref{thm:optorder}. For
MAC,\textcolor{black}{{} }satisfying the second condition implies that
the\textcolor{black}{{} weighted sum rate $\sum_{l}^{L}\tilde{\nu}_{l}\mathcal{I}_{l}\left(\mathbf{\Sigma}_{1:L},\mathbf{\Phi}\left(\pi\right)\right)$
is a concave function of $\mathbf{\Sigma}_{1:L}$, where $\tilde{\nu}_{l}$'s
are the polite water-filling levels of $\mathbf{\tilde{\Sigma}}_{1:L}$.
Then from the proof for the }sufficiency part of Theorem \textcolor{black}{\ref{thm:optimality},
$\mathbf{\tilde{\Sigma}}_{1:L}$ and }$\tilde{\pi}$ maximizes the
weighted sum-rate \textcolor{black}{$\sum_{l}^{L}\tilde{\nu}_{l}\mathcal{I}_{l}\left(\mathbf{\Sigma}_{1:L},\mathbf{\Phi}\left(\pi\right)\right)$
and thus achieves a }boundary point of the capacity region of MAC.
By Lemma \ref{lem:optodr}, \textcolor{black}{they achieve the global
optimality of }\textbf{\textcolor{black}{OFOP}}\textcolor{black}{{}
or }\textbf{\textcolor{black}{OSPMP}}\textcolor{black}{. The }sufficiency
part for BC follows from the rate duality.
\end{proof}

The proof of Theorem \ref{thm:optorder} suggests an algorithm to
improve the encoding/decoding order for each pseudo BC and pseudo
MAC by simply updating the encoding and decoding order according to
current polite water-filling levels. We refer to it as\emph{ }Algorithm
O and summarize it in Table \ref{tab:table8}.

\begin{table}
\caption{\label{tab:table8}Algorithm O (Improving the Encoding and Decoding
Order)}

\centering{}\begin{tabular}{l}
\hline
{\small Initialize the encoding and decoding order $\pi$ such that
$\mathbf{\Phi}\left(\pi\right)$ is valid.}\tabularnewline
{\small 1. Solve the }\textbf{\small FOP}{\small{} or }\textbf{\small SPMP}{\small{}
with fixed $\mathbf{\Phi}\left(\pi\right)$.}\tabularnewline
{\small 2. Calculate the polite water-filling levels }\textcolor{black}{\small from
the solution of }\textbf{\small FOP }{\small or}\tabularnewline
{\small $\;\;$$\;\;$}\textbf{\small SPMP}{\small{} obtained in step
1.}\tabularnewline
{\small 3. For each }\textcolor{black}{\small Pseudo }{\small BC }\textcolor{black}{\small and
Pseudo }{\small MAC with the polite water-filling }\tabularnewline
{\small $\;\;$$\;\;$levels}\textcolor{black}{\small{} obtained in
step 2, }{\small if $\pi$ satisfies Theorem \ref{thm:optorder},
output $\pi$ and}\tabularnewline
{\small $\;\;$$\;\;$stop. Otherwise, set $\pi$ to satisfy Theorem
\ref{thm:optorder} and return to step 1.}\tabularnewline
\end{tabular}
\end{table}

\begin{remrk}
For the special cases of MAC and BC, if Algorithm O converges, the
solution gives the optimal order by \textcolor{black}{Theorem }\ref{thm:optorder}.
\end{remrk}

\begin{remrk}
Optimal and sub-optimal algorithms for the special case of solving
\textbf{OSPMP} for SIMO MAC/MISO BC have been proposed in \cite{Yu_07TOC_MACorder}.
The optimal algorithm is much more complex than Algorithm O as it
involves several inner and outer iterations. The difference between
Algorithm O and the sub-optimal algorithm are as follows. 1) The sub-optimal
algorithm works for SIMO MAC, avoiding the calculation of beamforming
matrices, while Algorithm O works for MIMO cases; 2) To find the encoding/decoding
order, after obtain the optimal solution of \textbf{SPMP} with fixed
$\pi$ in step 1, the sub-optimal algorithm in \cite{Yu_07TOC_MACorder}
needs to solve an equation to obtain a weight vector $\left[w_{l}\right]_{,l=1,...,L}$
\textcolor{black}{such that $\left[\mathcal{I}_{1}^{0},\cdots,\mathcal{I}_{L}^{0}\right]$
is the optimal solution of the }\textbf{\textcolor{black}{WSRMP}}\textcolor{black}{{}
with the weights }$w_{l}$'s\textcolor{black}{{} and fixed }$\pi$\textcolor{black}{,
while in Algorithm O, the weight vector} is directly given by the
polite water-filling levels. Same as the algorithm in \cite{Yu_07TOC_MACorder},
it is possible for Algorithm O to cycle through a finite number of
orders. In this case, we can simply choose the best one from the finite
number of orders. For MAC and BC, it is observed that the reason of
non-convergence is that the corresponding boundary point can not be
achieved without time-sharing.
\end{remrk}

\begin{remrk}
Algorithms to solve \textbf{OFOP }and\textbf{ OSPMP }for MIMO BC have
been proposed in \cite{Nihar_06_ISIT_SymCapMIMODown} where the problems
are converted to the weighted sum-rate maximization problem. The algorithms
need to repeatedly solve a weighted sum-rate maximization problem
for weight vectors searched by the ellipsoid algorithm. The weighted
sum-rate maximization problem is solved by the steepest ascent algorithm.
For \textbf{OSPMP}, additional search for the total power is also
needed. In contrast, Algorithm O solves the problems directly in one
sequence of iterations, replaces the steepest ascent algorithm with
SINR based algorithms or polite water-filling, only searches weight
vector suggested by the polite water-filling level, resulting in much
lower complexity.
\end{remrk}

\begin{remrk}
\textcolor{black}{For BC and MAC, we can also use a simpler encoding/decoding
order in \cite{Liping_2007_Globecom_MEB} which is shown to be asymptotically
optimal for MAC (and thus also for BC by duality) when the target
rate of each user is the same and the number of users is large. At
each transmitter (receiver), the signal of the link with its channel
matrix having the $n^{\text{th}}$ smallest (largest) dominant singular
value is the $n^{\text{th}}$ one to be encoded (decoded). For convenience,
we refer to this order as MEB (maximum eigenmode beamforming) order.
It can be shown that MEB order is optimal for SISO MAC and BC.}
\end{remrk}

\section{\textcolor{black}{Simulation Results\label{sec:Simulation-Results}}}

Simulations are used to verify the performance of the proposed algorithms.\textcolor{black}{{}
Let each transmitter and receiver are equipped with $L_{T}$ and $L_{R}$
antennas respectively. }DPC and SIC are employed for interference
cancellation. Block fading channel is assumed and the channel matrices
are independently generated by $\mathbf{H}_{l,k}=\sqrt{g_{l,k}}\mathbf{H}_{l,k}^{\text{(W)}},\forall k,l$,
where $\mathbf{H}_{l,k}^{\text{(W)}}$ has zero-mean i.i.d. Gaussian
entries with unit variance; and $g_{l,k},\forall k,l$ is a positive
number and the value is set as $0\textrm{dB}$ except for Fig. \ref{fig:Fig3}
and Fig. \ref{fig:Fig4}. In \textcolor{black}{Fig. \ref{fig:Fig5}-\ref{fig:Fig6}},
each simulation is averaged over 100 random channel realizations\textcolor{black}{,
while in other figures, a single channel realization is considered.
In all simulations we conducted, Algorithm A, B, PR and PR1 are observed
to converge to a stationary point of the corresponding problem. We
call }\textit{\textcolor{black}{pseudo global optimum}}\textcolor{black}{{}
}the best solution among many solutions obtained by running the algorithm
for many times with different initial points and with the encoding/decoding
order obtained by Algorithm O. For the plots with iteration numbers,
we show rates or power after $x.5$ iterations/rounds, where the last
0.5 iteration/round is the forward link update.

\begin{figure}
\begin{centering}
\textsf{\includegraphics[clip,width=95mm]{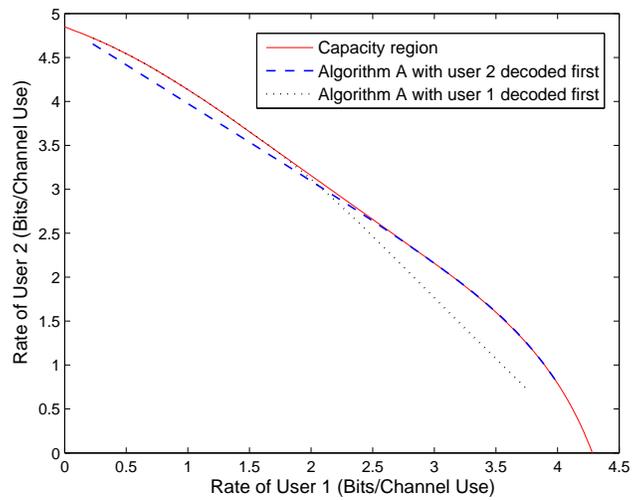}}
\par\end{centering}

\caption{\label{fig:Fig2}\textcolor{black}{Achieved rate region boundaries}
of a two-user MAC}

\end{figure}

\begin{figure}
\begin{centering}
\textsf{\includegraphics[clip,width=95mm]{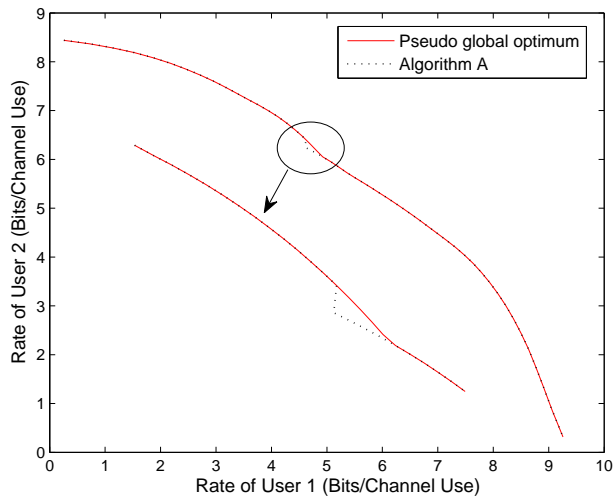}}
\par\end{centering}

\caption{\label{fig:Fig2a}\textcolor{black}{Achieved rate region} Boundary
of a two-user interference channel}

\end{figure}

\begin{figure}
\begin{centering}
\textsf{\includegraphics[clip,width=95mm]{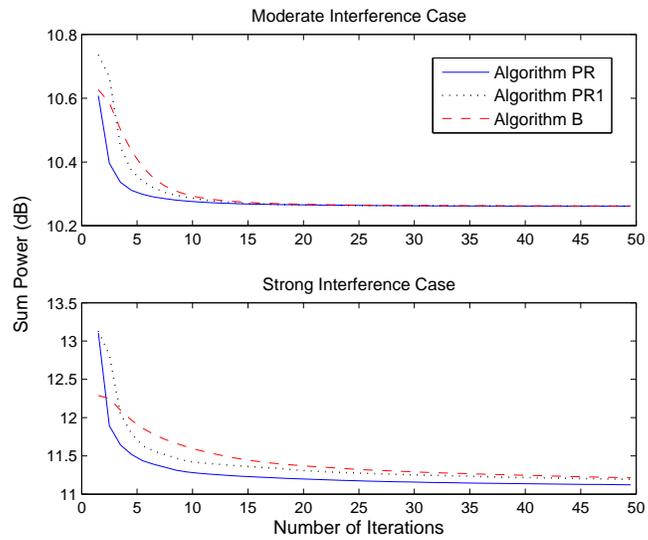}}
\par\end{centering}

\caption{\label{fig:Fig3}Convergence of the algorithms for the iTree network
\textcolor{black}{in Fig. \ref{fig:iTree}}}

\end{figure}

\begin{figure}
\begin{centering}
\textsf{\includegraphics[clip,width=95mm]{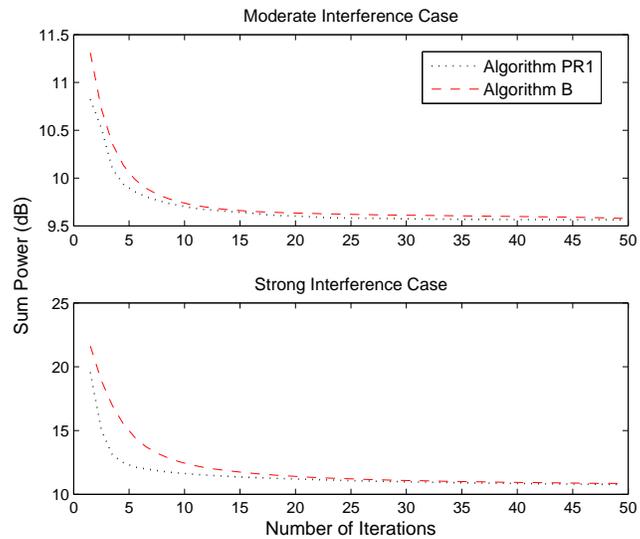}}
\par\end{centering}

\caption{\label{fig:Fig4}Convergence of the algorithms for a 3-user interference
channel}

\end{figure}

\begin{figure}
\begin{centering}
\textsf{\includegraphics[clip,width=95mm]{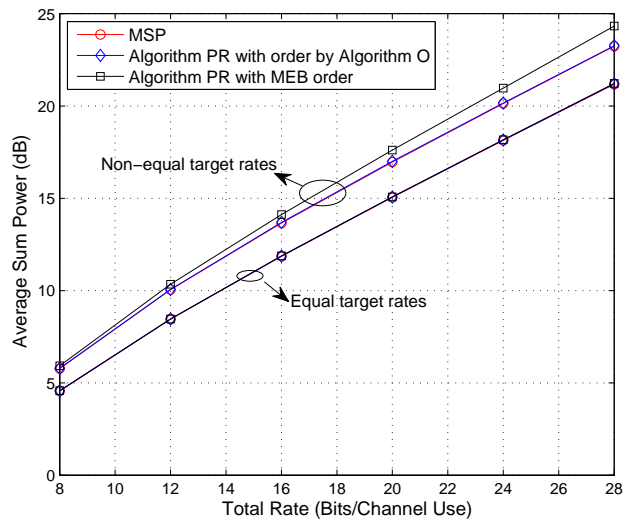}}
\par\end{centering}

\caption{\label{fig:Fig5}Average sum power vs. the total rate for a 4-user
MAC}

\end{figure}

\begin{figure}
\begin{centering}
\textsf{\includegraphics[clip,width=95mm]{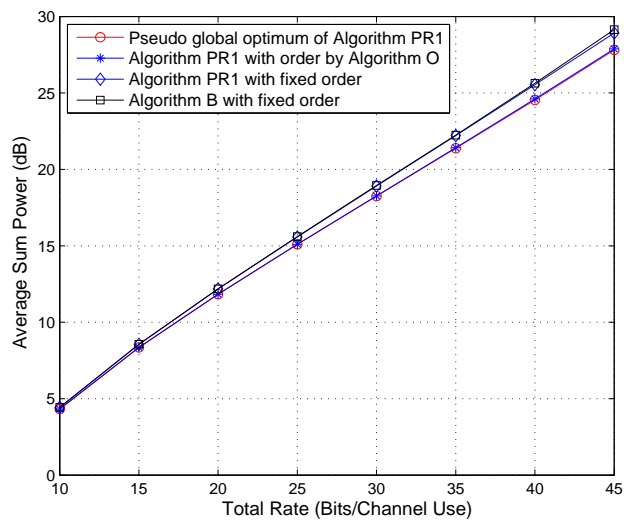}}
\par\end{centering}

\caption{\label{fig:Fig5a}Average sum power vs. the total rate for\textcolor{black}{{}
the B-MAC network in Fig. \ref{fig:sysFig1}}}

\end{figure}

\begin{figure}
\begin{centering}
\textsf{\includegraphics[clip,width=95mm]{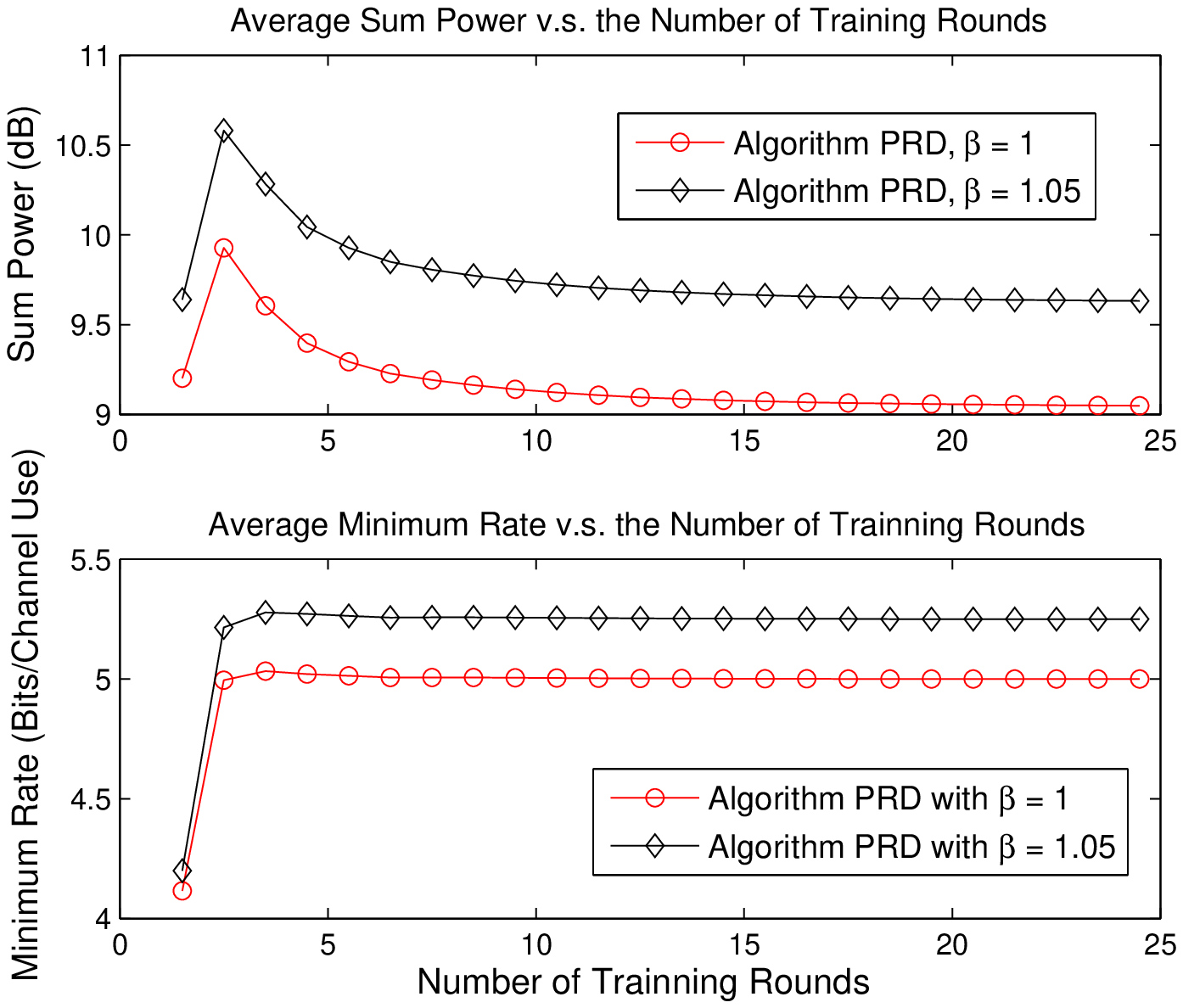}}
\par\end{centering}

\caption{\label{fig:Fig6}Convergence of the distributed algorithm for a 3-user
interference channel}

\end{figure}

\textcolor{black}{Algorithm A can be used to find the achievable rate
region boundary by varying the target rates $\mathcal{I}_{l}^{0}$'s.
It finds the point where the boundary is intersected by the ray $\alpha\left(\mathcal{I}_{1}^{0},\cdots,\mathcal{I}_{L}^{0}\right)$.
In Fig. \ref{fig:Fig2}, we plot the boundaries of the rate regions
achieved by Algorithm A with different decoding orders for a two-user
MAC with $L_{T}=2,\ L_{R}=4$. It can be observed that the convex
hull of the rate regions achieved by Algorithm A is the same as the
capacity region, which implies that Algorithm A does achieve the global
optimum for this case, and thus is a low complexity approach to calculate
the capacity region for MIMO MAC.}%
{}\textcolor{black}{{} In Fig. \ref{fig:Fig2a}, we plot the boundary
of the rate region achieved by Algorithm A for a two-user interference
channel with $L_{T}=2,\ L_{R}=4$. The }pseudo optimum boundary\textcolor{black}{{}
achieved by Algorithm A is also plotted for comparison. }In most places,
the two boundaries overlap with each other except for a small area.
It demonstrates that Algorithm A can find good solutions even with
a single initial point.

We demonstrate the superior convergence speed of Algorithm PR and
PR1. Sum power versus iteration number is shown in Fig. \ref{fig:Fig3}
for the iTree network of Fig. \ref{fig:iTree} and in Fig. \ref{fig:Fig4}
for a 3-user interference channel. Each node has four antennas. \textcolor{black}{The
target rate for each link is set as 5 bits/channel use.} In the upper
sub-plot of Fig. \ref{fig:Fig3}, we consider the moderate interference
case, where we set $g_{l,k}=0\textrm{dB},\ \forall k,l$. In the lower
sub-plot of Fig. \ref{fig:Fig3}, we consider strong interference
case, where we set $g_{l,3}=10\textrm{dB},\ l=1,2$ for the interfering
links, and $g_{l,k}=0\textrm{dB}$ for other $k,l$'s. It is not surprising
that Algorithms PR and PR1 have faster convergence speed because polite
water-filling exploits the structure of the problem. In the upper
sub-plot of Fig. \ref{fig:Fig4}, we set $g_{l,k}=0\textrm{dB},\forall k,l$.
In the lower sub-plot, we consider a strong interference channel,
i.e., $g_{l,k}=10\textrm{dB},\forall k\neq l$, and $g_{l,k}=0\textrm{dB},\forall k=l$.
Again, Algorithm PR1 converges faster than Algorithm B. Since the
problem is non-convex, the algorithms may converge to different stationary
points. But it can be observed that all the stationary points achieve
similar performance.

\textcolor{black}{In Fig. \ref{fig:Fig5}, we evaluate the performance
of Algorithm PR for a 4-user MAC with $L_{T}=2,\ L_{R}=4$. The }`\textcolor{black}{MSP}'\textcolor{black}{{}
is the optimal solution obtained by the }`\textcolor{black}{Algorithm
2}'\textcolor{black}{{} in \cite{Nihar_06_ISIT_SymCapMIMODown}, which
has much higher complexity as discussed in Section \ref{sub:Soltuions-for-MAC}.
For Algorithm PR, both the order obtained by Algorithm O and the MEB
order is considered. When the target rate for each user is the same,
Algorithm PR with both decoding orders achieves nearly the same sum
power as the MSP but with much lower complexity. When the target rates
of the users are different and are }set as $\left[R_{s},2R_{s},4R_{s},8R_{s}\right]/15$,
where $R_{s}$ is the total required rate, \textcolor{black}{Algorithm
PR} with \textcolor{black}{the order} \textcolor{black}{obtained by
Algorithm O} still achieves near-optimal performance, while \textcolor{black}{Algorithm
PR} with the MEB order performs worse than that. Not showing is that\textcolor{black}{{}
Algorithm B and PR1 also achieve the same performance as Algorithm
PR. In Fig. \ref{fig:Fig5a}, we evaluate the performance of Algorithms
B and PR1 for the B-MAC network in Fig. \ref{fig:sysFig1} with $L_{T}=L_{R}=4$.
}The target rates are set as $\left[R_{s},R_{s},2R_{s},4R_{s},8R_{s}\right]/16$.
The encoding/decoding order is partially fixed and is the same as
that in Example \ref{exa:pseudoBM} of Section \ref{sub:OrderOptimization}.
For the pseudo MAC formed by link 2 and link 3 and the pseudo BC formed
by link 4 and link 5, the fixed order that $\mathbf{x}_{3}$ is decoded
after $\mathbf{x}_{2}$ and $\mathbf{x}_{5}$ is decoded after $\mathbf{x}_{4}$,
and its improved order obtained by Algorithm O are applied. \textcolor{black}{With
the }improved order obtained by Algorithm O\textcolor{black}{, the
performance of Algorithm B is not shown because both algorithms achieve
nearly the same performance as the Pseudo global optimum of Algorithm
PR1, while the algorithms with the fixed order suffers a performance
loss.}

We illustrate the convergence behavior of the distributed optimization
with local CSI. Fig. \ref{fig:Fig6} plots the total transmit power
and the minimum rate of the users achieved by Algorithm PRD versus
the number of training rounds for a 3-user interference channel with
\textcolor{black}{$L_{T}=L_{R}=4$. In Algorithm PRD, the target rates
are set as $\beta\left(\mathcal{I}_{1}^{0},\mathcal{I}_{2}^{0},\mathcal{I}_{3}^{0}\right)$
with $\beta=1$ and $\beta=1.05$ respectively, where $\mathcal{I}_{l}^{0}=5$
(bits/channel use), $l=1,2,3$. It can be observed that after 2.5
rounds the rates are close to the targets and after $3.5$ rounds,
the powers are also close to the that of infinite rounds}. When\textcolor{black}{{}
$\beta=1$, the achieved minimum rate after $3.5$ rounds equals to
or exceeds 5 (bits/channel use) in 88 out of 100 simulations. }When\textcolor{black}{{}
$\beta=1.05$, achieved minimum rate after $3.5$ rounds equals to
or exceeds 5 (bits/channel use) in all 100 simulations, while the
total transmit power is about 0.7 dB larger. This suggests a trick
that use higher target rates than true targets in order to satisfy
the target rates in fewer number of iterations at the expense of more
power.}

\section{\textcolor{black}{Conclusion\label{sec:Conclusion}}}

The general MIMO one-hop interference networks named B-MAC networks
with Gaussian input and any valid coupling matrices are considered.
We design algorithms for maximizing the minimum of weighted rates
under sum power constraints and for minimizing sum power under rate
constraints. They can be used in \textcolor{black}{admission control
and }in \textcolor{black}{guaranteeing the quality of service.} Two
kinds of algorithms are designed. The first kind takes advantage of
existing SINR optimization algorithms by finding simple and optimal
mappings between the achievable rate region and the SINR region. The
mappings can be used for many other optimization problems. The second
kind takes advantage of the\textcolor{black}{{} polite water-filling
structure of the optimal input found in \cite{Liu_IT10s_Duality_BMAC}}.
Both centralized and distributed algorithms are designed.\textcolor{black}{{}
The proposed algorithms are either proved or shown by simulations
to converge to a stationary point, which may not be optimal for non-convex
cases, but is shown by simulations to be good solutions.}

\appendix

\subsection{\textcolor{black}{Proof for Theorem \ref{thm:EquSINRopt}\label{sub:Proof-for-EquSINRopt}}}

\textcolor{black}{Because each link is equivalent to a single-user
channel after whitening the interference plus noise, we only need
to prove that for any $\mathbf{\Sigma}$ achieving a rate $\mathcal{I}\left(\mathbf{\Sigma}\right)=\textrm{log}\left|\mathbf{I}+\mathbf{H\Sigma H^{\dagger}}\right|$
in a single-user channel $\mathbf{H}$, there exists a decomposition
of $\mathbf{\Sigma}=\dot{\mathbf{T}}\dot{\mathbf{T}}^{\dagger}$ leading
to $\left\{ \mathbf{T},\mathbf{R},\mathbf{p}\right\} $ which achieves
a set of SINRs $\gamma_{m}=e^{\mathcal{I}\left(\mathbf{\Sigma}\right)/M}-1,m=1,...,M$.}

\textcolor{black}{First, we show that considering unitary precoding
matrix $\mathbf{V}\in\mathbb{C}^{M\times M}$ will not loss generality.
Note that $\mathcal{I}\left(\mathbf{\Sigma}\right)=\textrm{log}\left|\mathbf{I}+\mathbf{H}\dot{\mathbf{T}}\mathbf{VV^{\dagger}}\dot{\mathbf{T}}^{\dagger}\mathbf{H^{\dagger}}\right|=\textrm{log}\left|\mathbf{I}+\mathbf{\bar{H}VV^{\dagger}\bar{H}^{\dagger}}\right|,$
where $\mathbf{\bar{H}}=\mathbf{H}\dot{\mathbf{T}}$ is the equivalent
channel with unitary precoding matrix $\mathbf{V}=\left[\mathbf{v}_{1},...,\mathbf{\mathbf{v}}_{M}\right]$.
Define \[
\bar{\mathbf{A}}_{m}=\bar{\mathbf{H}}^{\dagger}\left(\sum_{i=m+1}^{M}\mathbf{\bar{H}}\mathbf{v}_{i}\mathbf{v}_{i}^{\dagger}\mathbf{\bar{H}}^{\dagger}+\mathbf{I}\right)^{-1}\mathbf{\bar{H}}.\]
The SINR of the $m^{th}$ stream achieved by the MMSE-SIC receiver
is given by \cite{Varanasi_Asilomar97_MMSE_is_optimal} \[
\gamma_{m}=\mathbf{v}_{m}^{\dagger}\bar{\mathbf{A}}_{m}\mathbf{v}_{m}.\]
Hence, we only need to find a unitary precoding matrix $\mathbf{V}$
such that $\mathbf{v}_{m}^{\dagger}\bar{\mathbf{A}}_{m}\mathbf{v}_{m}=e^{\mathcal{I}\left(\mathbf{\Sigma}\right)/M}-1,m=1,...,M$.
Then the precoding matrix for the original channel is give by $\dot{\mathbf{T}}^{'}=\dot{\mathbf{T}}\mathbf{V}$. }

\textcolor{black}{We will use the method of induction.}%
{}\textcolor{black}{{} We first find a unit vector $\mathbf{v}_{M}$ such
that $\mathbf{v}_{M}^{\dagger}\bar{\mathbf{A}}_{M}\mathbf{v}_{M}=e^{\mathcal{I}\left(\mathbf{\Sigma}\right)/M}-1$.
Let $\lambda_{i}^{(M)}$ be the $i^{th}$ largest eigenvalue of $\bar{\mathbf{A}}_{M}$
and $\mathbf{u}_{i}^{(M)}$ be the corresponding eigenvector. Since
$\mathcal{I}\left(\mathbf{\Sigma}\right)=\textrm{log}\left|\mathbf{I}+\mathbf{\bar{H}VV^{\dagger}\bar{H}^{\dagger}}\right|=\textrm{log}\left|\mathbf{I}+\bar{\mathbf{A}}_{M}\right|=\textrm{log}\prod_{i=1}^{M}\left(1+\lambda_{i}^{(M)}\right)$,
we must have $\lambda_{1}^{(M)}\geq e^{\mathcal{I}\left(\mathbf{\Sigma}\right)/M}-1$
and $\lambda_{M}^{(M)}\leq e^{\mathcal{I}\left(\mathbf{\Sigma}\right)/M}-1$.
Note that $\mathbf{v}_{M}^{\dagger}\bar{\mathbf{A}}_{M}\mathbf{v}_{M}=\sum_{i=1}^{M}\left|\mathbf{v}_{M}^{\dagger}\mathbf{u}_{i}^{(M)}\right|^{2}\lambda_{i}^{(M)}.$
Because $\left\{ \mathbf{u}_{i}^{(m)},i=1,...,M\right\} $ form orthogonal
bases, there exists a $\mathbf{v}_{M}$ such that \begin{align*}
\left|\mathbf{v}_{M}^{\dagger}\mathbf{u}_{i}^{(M)}\right|^{2} & =0,\ i=2,...,M-1,\\
\left|\mathbf{v}_{M}^{\dagger}\mathbf{u}_{1}^{(M)}\right|^{2}\lambda_{1}^{(M)}+ & \left|\mathbf{v}_{M}^{\dagger}\mathbf{u}_{1}^{(M)}\right|^{2}\lambda_{M}^{(M)}=e^{\mathcal{I}\left(\mathbf{\Sigma}\right)/M}-1.\end{align*}
Then it follows $\mathbf{v}_{M}^{\dagger}\bar{\mathbf{A}}_{M}\mathbf{v}_{M}=e^{\mathcal{I}\left(\mathbf{\Sigma}\right)/M}-1$. }

\textcolor{black}{Assume we already found a set of mutual orthogonal
unit vectors $\mathbf{v}_{l},l=m+1,...,M$ such that $\mathbf{v}_{l}^{\dagger}\bar{\mathbf{A}}_{l}\mathbf{v}_{l}=e^{\mathcal{I}\left(\mathbf{\Sigma}\right)/M}-1,l=m+1,...,M$.
The rest is to prove that there exists a $\mathbf{v}_{m}$ such that
$\mathbf{v}_{m}^{\dagger}\bar{\mathbf{A}}_{m}\mathbf{v}_{m}=e^{\mathcal{I}\left(\mathbf{\Sigma}\right)/M}-1$
and $\mathbf{v}_{m}$ is orthogonal to $\mathbf{v}_{l},l=m+1,...,M$.
Perform SVD $\bar{\mathbf{A}}_{m}=\mathbf{U}_{m}\mathbf{D}_{m}\mathbf{U}_{m}^{\dagger}$.
Let $\lambda_{n}^{(m)}$ be the $n^{th}$ largest eigenvalue of $\bar{\mathbf{A}}_{m}$
and $\mathbf{u}_{n}^{(m)}$ be the corresponding eigenvector. Define
$\mathbf{\hat{u}}_{n}^{(m)}=\mathbf{u}_{n}^{(m)}-\sum_{j=m+1}^{M}\mathbf{v}_{j}^{\dagger}\mathbf{u}_{n}^{(m)}\mathbf{v}_{j},n=1,...,M$,
$\mathbf{\hat{U}}_{m}=\left[\mathbf{\hat{u}}_{1}^{(m)},...,\mathbf{\hat{u}}_{M}^{(m)}\right]$
and $\mathbf{\tilde{A}}_{m}=\mathbf{\hat{U}}_{m}\mathbf{D}_{m}\mathbf{\hat{U}}_{m}^{\dagger}$.
Then for $i,j=1,...,m$, we have \begin{align}
\mathbf{v}_{i}^{\dagger}\bar{\mathbf{A}}_{m}\mathbf{v}_{j}= & \mathbf{v}_{i}^{\dagger}\sum_{n=1}^{M}\mathbf{u}_{n}^{(m)}\lambda_{n}^{(m)}\left(\mathbf{u}_{n}^{(m)}\right)^{\dagger}\mathbf{v}_{j}\nonumber \\
= & \mathbf{v}_{i}^{\dagger}\sum_{n=1}^{M}\mathbf{\hat{u}}_{n}^{(m)}\lambda_{n}^{(m)}\left(\mathbf{\hat{u}}_{n}^{(m)}\right)^{\dagger}\mathbf{v}_{j}\label{eq:vmAMVMH}\\
= & \mathbf{v}_{i}^{\dagger}\mathbf{\tilde{A}}_{m}\mathbf{v}_{j},\nonumber \end{align}
where (\ref{eq:vmAMVMH}) follows from the definition of $\mathbf{\hat{u}}_{n}^{(m)}$
and the fact that $\mathbf{v}_{i}^{\dagger}\mathbf{v}_{k}=0,$ $i=1,...,m$,
$k=m+1,...,M$. Because $\mathbf{\tilde{A}}_{m}$ is positive semidefinite
and \begin{equation}
\mathbf{\tilde{A}}_{m}\mathbf{v}_{i}=\mathbf{0},i=m+1,...,M,\label{eq:Amvizero}\end{equation}
the rank of $\mathbf{\tilde{A}}_{m}$ must be less than $m+1$. Let
$\tilde{\lambda}_{i}^{(m)}$ be the $i^{th}$ largest eigenvalue of
$\mathbf{\tilde{A}}_{m}$ and $\mathbf{\tilde{u}}_{i}^{(m)}$ be the
corresponding eigenvector. Then we must have $\tilde{\lambda}_{i}^{(m)}=0,i=m+1,...,M$.
Define $\mathbf{V}_{m}=\left[\mathbf{v}_{1},...,\mathbf{\mathbf{v}}_{m}\right]$.
Note that the interference from the last $M-m$ streams is $\sum_{i=m+1}^{M}\mathbf{\bar{H}}\mathbf{v}_{i}$
. Then the sum rate of the first $m$ streams is given by\begin{align}
 & \textrm{log}\left|\mathbf{I}+\mathbf{\bar{H}}\mathbf{V}_{m}\mathbf{V}_{m}^{\dagger}\mathbf{\bar{H}^{\dagger}}\left(\sum_{i=m+1}^{M}\mathbf{\bar{H}}\mathbf{v}_{i}\mathbf{v}_{i}^{\dagger}\mathbf{\bar{H}}^{\dagger}+\mathbf{I}\right)^{-1}\right|\nonumber \\
= & \textrm{log}\left|\mathbf{I}+\mathbf{V}_{m}^{\dagger}\bar{\mathbf{A}}_{m}\mathbf{V}_{m}\right|\nonumber \\
= & \textrm{log}\left|\mathbf{I}+\mathbf{V}_{m}^{\dagger}\mathbf{\tilde{A}}_{m}\mathbf{V}_{m}\right|\label{eq:SRM1}\\
= & \textrm{log}\left|\mathbf{I}+\mathbf{V}^{\dagger}\mathbf{\tilde{A}}_{m}\mathbf{V}\right|\label{eq:SRM2}\\
= & \textrm{log}\prod_{i=1}^{m}\left(1+\tilde{\lambda}_{i}^{(m)}\right)=\frac{m}{M}\mathcal{I}\left(\mathbf{\Sigma}\right),\nonumber \end{align}
where (\ref{eq:SRM1}) and (\ref{eq:SRM2}) follows from (\ref{eq:vmAMVMH})
and (\ref{eq:Amvizero}) respectively. Therefore we must have $\tilde{\lambda}_{1}^{(m)}\geq e^{\mathcal{I}\left(\mathbf{\Sigma}\right)/M}-1$
and $\tilde{\lambda}_{m}^{(m)}\leq e^{\mathcal{I}\left(\mathbf{\Sigma}\right)/M}-1$.
Note that \begin{align*}
\mathbf{v}_{m}^{\dagger}\bar{\mathbf{A}}_{m}\mathbf{v}_{m} & =\mathbf{v}_{m}^{\dagger}\mathbf{\tilde{A}}_{m}\mathbf{v}_{m}\\
 & =\sum_{n=1}^{M}\left|\mathbf{v}_{m}^{\dagger}\mathbf{\tilde{u}}_{n}^{(m)}\right|^{2}\tilde{\lambda}_{n}^{(m)}.\end{align*}
Because $\left\{ \mathbf{\tilde{u}}_{i}^{(m)},i=1,...,m\right\} $
form orthogonal bases of the $m$-dimensional subspace orthogonal
to $\mathbf{v}_{l},l=m+1,...,M$, there exits a unit vector $\mathbf{v}_{m}$
in this subspace such that \begin{align*}
\left|\mathbf{v}_{m}^{\dagger}\tilde{\mathbf{u}}_{i}^{(m)}\right|^{2} & =0,i=2,...,m-1,m+1,...,M,\\
\left|\mathbf{v}_{m}^{\dagger}\mathbf{\tilde{u}}_{1}^{(m)}\right|^{2}\tilde{\lambda}_{1}^{(m)}+ & \left|\mathbf{v}_{m}^{\dagger}\mathbf{\tilde{u}}_{m}^{(m)}\right|^{2}\tilde{\lambda}_{m}^{(m)}=e^{\mathcal{I}\left(\mathbf{\Sigma}\right)/M}-1.\end{align*}
Then we have $\mathbf{v}_{m}^{\dagger}\bar{\mathbf{A}}_{m}\mathbf{v}_{m}=e^{\mathcal{I}\left(\mathbf{\Sigma}\right)/M}-1$.
This completes the proof.}

\subsection{\textcolor{black}{Proof for Theorem \ref{thm:equpoweropt}\label{sub:Proof-for-Theorem-Equpow}}}

\textcolor{black}{Note that $\mathbf{\Sigma}=\sum_{m=1}^{M}p_{m}\mathbf{t}_{m}\mathbf{t}_{m}^{\dagger}$
implies that $M\ge\text{Rank}(\mathbf{\Sigma})$. Define an $M\times M$
DFT matrix $\mathbf{F}$ where the element at the $k^{th}$ row and
$l^{th}$ column is $\mathbf{F}_{k,l}=e^{-\frac{2\pi kl}{M}j}/\sqrt{M}$.
If $M$ is chosen to be greater than or equal to $L_{T}$, let $\mathbf{F}_{0}\in\mathbb{C}^{L_{T}\times M}$
be the matrix comprised of the first $L_{T}$ rows of $\mathbf{F}$.
Otherwise, let $\mathbf{F}_{0}\in\mathbb{C}^{L_{T}\times M}$ be the
matrix such that the upper sub matrix are $\mathbf{F}$, and other
elements are zero. Perform SVD $\mathbf{\Sigma}=\mathbf{U}\mathbf{D}\mathbf{U}^{\dagger}$,
where the diagonal elements of $\mathbf{D}$ are positive and in descending
order. Let $\dot{\mathbf{T}}=\mathbf{U}\mathbf{D}^{1/2}\mathbf{F}_{0}$.
It can be verified that $\dot{\mathbf{T}}\dot{\mathbf{T}}^{\dagger}=\mathbf{\Sigma}$.
The norms of the columns of $\dot{\mathbf{T}}$ are the diagonal elements
of $\dot{\mathbf{T}}^{\dagger}\dot{\mathbf{T}}=\mathbf{F}_{0}^{\dagger}\mathbf{D}\mathbf{F}_{0}$
and they are equal to $\frac{\sum_{i=1}^{L_{T}}\mathbf{D}_{i,i}}{M}$.
Then the corresponding transmit powers satisfy $p_{m}=\textrm{Tr}\left(\mathbf{\mathbf{\Sigma}}\right)/M,\forall m$.}

\subsection{\textcolor{black}{Proof of the Theorem \ref{thm:optimality}\label{sub:Proof-of-optcon}}}

\textcolor{black}{First, we show that any $\tilde{\mathbf{\Sigma}}_{1:L}$
satisfying the optimality conditions for }\textbf{\textcolor{black}{FOPa}}\textcolor{black}{{}
must satisfy the KKT conditions of }\textbf{\textcolor{black}{FOPa}}\textcolor{black}{.
The Lagrangian of }\textbf{\textcolor{black}{FOPa}}\textcolor{black}{{}
(\ref{eq:P1a}) is\begin{eqnarray}
 & L\left(\mathbf{\Sigma}_{1:L},\nu_{l},\mathbf{\Theta}_{1:L}\right)\label{eq:LagSPMP}\\
= & \frac{\left(1-\sum_{l\neq1}\nu_{l}\right)}{\mathcal{I}_{1}^{0}}\mathcal{I}_{1}\left(\mathbf{\Sigma}_{1:L},\mathbf{\Phi}\right)+\frac{\sum_{l\neq1}\nu_{l}}{\mathcal{I}_{l}^{0}}\mathcal{I}_{l}\left(\mathbf{\Sigma}_{1:L},\mathbf{\Phi}\right)\nonumber \\
 & +\mu\left(P_{T}-\sum_{l=1}^{L}\textrm{Tr}\left(\mathbf{\Sigma}_{l}\right)\right)+\sum_{l=1}^{L}\textrm{Tr}\left(\mathbf{\Sigma}_{l}\mathbf{\Theta}_{l}\right).\nonumber \end{eqnarray}
The KKT conditions are\begin{eqnarray}
 & \nabla_{\mathbf{\Sigma}_{l}}L=0,\ \textrm{Tr}\left(\mathbf{\Sigma}_{l}\mathbf{\Theta}_{l}\right)=0,\ \mathbf{\Sigma}_{l},\mathbf{\Theta}_{l}\succeq0,\ \forall l;\nonumber \\
 & \nu_{l}\geq0,\ \mathcal{I}_{1}^{0}\mathcal{I}_{l}\left(\mathbf{\Sigma}_{1:L},\mathbf{\Phi}\right)=\mathcal{I}_{l}^{0}\mathcal{I}_{1}\left(\mathbf{\Sigma}_{1:L},\mathbf{\Phi}\right),\ \forall l\neq1;\label{eq:KKT1}\\
 & \mu\geq0,\ P_{T}=\sum_{l=1}^{L}\textrm{Tr}\left(\mathbf{\Sigma}_{l}\right).\nonumber \end{eqnarray}
Recall $\tilde{\nu}_{l}$ is the polite water-filling level in the
optimality condition. Let $\mu=1/\sum_{l=1}^{L}\mathcal{I}_{l}^{0}\tilde{\nu}_{l}$
and $\nu_{l}=\mu\mathcal{I}_{l}^{0}\tilde{\nu}_{l}$. Then the condition
$\nabla_{\mathbf{\Sigma}_{l}}L=0$ can be expressed as \begin{align}
 & \sum_{k\neq l}\tilde{\nu}_{k}\mathbf{\Phi}_{k,l}\mathbf{H}_{k,l}^{\dagger}\left(\mathbf{\Omega}_{k}^{-1}-\left(\mathbf{\Omega}_{k}+\mathbf{H}_{k,k}\mathbf{\Sigma}_{k}\mathbf{H}_{k,k}^{\dagger}\right)^{-1}\right)\mathbf{H}_{k,l}+\mathbf{I}\nonumber \\
 & =\tilde{\nu}_{l}\mathbf{H}_{l,l}^{\dagger}\left(\mathbf{\Omega}_{l}+\mathbf{H}_{l,l}\mathbf{\Sigma}_{l}\mathbf{H}_{l,l}^{\dagger}\right)^{-1}\mathbf{H}_{l,l}+\frac{1}{\mu}\mathbf{\Theta}_{l}.\label{eq:DLequ}\end{align}
By Theorem \ref{thm:FequRGWF}, $\mathbf{\tilde{\hat{\mathbf{\Sigma}}}}_{1:L}$
can be expressed as\begin{equation}
\mathbf{\tilde{\hat{\mathbf{\Sigma}}}}_{l}=\tilde{\nu}_{l}\left(\tilde{\mathbf{\Omega}}_{l}^{-1}-\left(\tilde{\mathbf{\Omega}}_{l}+\mathbf{H}_{l,l}\tilde{\mathbf{\Sigma}}_{l}\mathbf{H}_{l}^{\dagger}\right)^{-1}\right),\forall l.\label{eq:SigmahM}\end{equation}
Substitute $\mathbf{\tilde{\mathbf{\Sigma}}}_{1:L}$ into condition
(\ref{eq:DLequ}) to obtain \begin{equation}
\tilde{\hat{\mathbf{\Omega}}}_{l}=\tilde{\nu}_{l}\mathbf{H}_{l,l}^{\dagger}\left(\tilde{\mathbf{\Omega}}_{l}+\mathbf{H}_{l,l}\tilde{\mathbf{\Sigma}}_{l}\mathbf{H}_{l,l}^{\dagger}\right)^{-1}\mathbf{H}_{l,l}+\frac{1}{\mu}\mathbf{\Theta}_{l},\ \forall l.\label{eq:KeyKKT}\end{equation}
Because the KKT condition (\ref{eq:KeyKKT}) is also that of the single
user polite water-filling problem over the channel $\tilde{\mathbf{\Omega}}_{l}^{-1/2}\mathbf{H}_{l,l}\tilde{\hat{\mathbf{\Omega}}}_{l}^{-1/2}$
and by the optimality condition, $\mathbf{\tilde{\Sigma}}_{l}$ has
polite water-filling structure over this channel, $\mathbf{\tilde{\mathbf{\Sigma}}}_{1:L}$
satisfies condition (\ref{eq:KeyKKT}). It can be verified that $\mathbf{\tilde{\mathbf{\Sigma}}}_{1:L}$
also satisfies all other KKT conditions in (\ref{eq:KKT1}).}

\textcolor{black}{With a similar proof as above, one can show that
for }\textbf{\textcolor{black}{SPMP}}\textcolor{black}{, the necessary
optimality conditions also implies that the KKT conditions hold.}

\textcolor{black}{The sufficient part for }\textbf{\textcolor{black}{FOPa}}\textcolor{black}{{}
is proved by showing that the optimum of }\textbf{\textcolor{black}{FOPa}}\textcolor{black}{{}
is equal to the optimum of some weighted sum rate maximization problem
(}\textbf{\textcolor{black}{WSRMP}}\textcolor{black}{) and the KKT
conditions are sufficient for global optimality when the weighted
sum rate $\sum_{l}^{L}\tilde{\nu}_{l}\mathcal{I}_{l}\left(\mathbf{\mathbf{\Sigma}}_{1:L},\mathbf{\Phi}\right)$
is a convex function. Suppose certain $\tilde{\mathbf{\Sigma}}_{1:L}$
satisfies the optimality conditions for }\textbf{\textcolor{black}{FOPa}}\textcolor{black}{.
Then following similar steps from (\ref{eq:LagSPMP}) to (\ref{eq:KeyKKT}),
it can be shown that $\tilde{\mathbf{\Sigma}}_{1:L}$ satisfies the
KKT conditions of the following }\textbf{\textcolor{black}{WSRMP}}\textcolor{black}{.}

\textcolor{black}{\begin{align}
\textrm{\textbf{WSRMP}:}\underset{\Sigma_{1:L}}{\textrm{max}} & \sum_{l=1}^{L}\tilde{\nu}_{l}\mathcal{I}_{l}\left(\mathbf{\Sigma}_{1:L},\mathbf{\Phi}\right)\label{eq:WSRMP}\\
\textrm{s.t.} & \mathbf{\Sigma}_{l}\succeq0,l=1,\cdots,L\:\textrm{and}\:\sum_{l=1}^{L}\textrm{Tr}\left(\mathbf{\Sigma}_{l}\right)\leq P_{T},\nonumber \end{align}
where $\left\{ \tilde{\nu}_{l}\right\} $ are the polite water-filling
levels corresponding to $\tilde{\mathbf{\Sigma}}_{1:L}$.}

\textcolor{black}{Because $\sum_{l=1}^{L}\tilde{\nu}_{l}\mathcal{I}_{l}\left(\mathbf{\Sigma}_{1:L},\mathbf{\Phi}\right)$
is a concave function of $\mathbf{\Sigma}_{1:L}$, the KKT conditions
are sufficient for the global optimality of }\textbf{\textcolor{black}{WSRMP}}\textcolor{black}{.
Noting that $\mathcal{I}_{l}\left(\tilde{\mathbf{\Sigma}}_{1:L},\mathbf{\Phi}\right)=\alpha\mathcal{I}_{l}^{0},\forall l$.
If $\mathbf{\tilde{\mathbf{\Sigma}}}_{1:L}$ is not an optimum of
}\textbf{\textcolor{black}{FOPa}}\textcolor{black}{, there exists
a $\mathbf{\tilde{\mathbf{\Sigma}}}_{1:L}^{'}$ satisfying all the
constraints, such that $\mathcal{I}_{1}\left(\mathbf{\tilde{\mathbf{\Sigma}}}_{1:L}^{'},\mathbf{\Phi}\right)/\mathcal{I}_{1}^{0}>\alpha$,
from which it follows that $\mathcal{I}_{l}\left(\mathbf{\tilde{\mathbf{\Sigma}}}_{1:L}^{'},\mathbf{\Phi}\right)/\mathcal{I}_{l}^{0}>\alpha,\forall l$
since $\mathbf{\tilde{\mathbf{\Sigma}}}_{1:L}^{'}$ satisfy the first
constraint in }\textbf{\textcolor{black}{FOPa}}\textcolor{black}{.
Then we must have $\sum_{l=1}^{L}\tilde{\nu}_{l}\mathcal{I}_{l}\left(\mathbf{\tilde{\mathbf{\Sigma}}}_{l}^{'},\mathbf{\Phi}\right)>\sum_{l=1}^{L}\tilde{\nu}_{l}\alpha\mathcal{I}_{l}^{0}=\sum_{l=1}^{L}\tilde{\nu}_{l}\mathcal{I}_{l}\left(\tilde{\mathbf{\Sigma}}_{1:L},\mathbf{\Phi}\right)$,
which contradicts the fact that $\mathbf{\tilde{\mathbf{\Sigma}}}_{1:L}$
is an optimum of }\textbf{\textcolor{black}{WSRMP}}\textcolor{black}{.}

\textcolor{black}{Similarly, using KKT conditions, the sufficient
part for }\textbf{\textcolor{black}{SPMP}}\textcolor{black}{{} can be
proved by showing that $\tilde{\mathbf{\Sigma}}_{1:L}$ is the global
optimum of the following convex optimization problem }

\textcolor{black}{\begin{align*}
\underset{\Sigma_{1:L}}{\textrm{min}} & \sum_{l=1}^{L}\textrm{Tr}\left(\mathbf{\Sigma}_{l}\right)\\
\textrm{s.t.} & \mathbf{\Sigma}_{l}\succeq0,l=1,\cdots,L\:\textrm{and}\:\sum_{l=1}^{L}\tilde{\nu}_{l}\mathcal{I}_{l}\left(\mathbf{\Sigma}_{1:L},\mathbf{\Phi}\right)\geq\sum_{l=1}^{L}\tilde{\nu}_{l}\mathcal{I}_{l}^{0}.\end{align*}
}


\end{document}